\documentclass{amsart}
\usepackage{amssymb}
\usepackage[frenchb,english]{babel}
\usepackage[T1]{fontenc}
\usepackage{amsmath}
\usepackage{amsfonts}
\usepackage{lmodern}
\usepackage[left=2cm,right=2cm,top=2cm,bottom=2cm]{geometry}
\usepackage{variations}
\usepackage{graphicx} 
\usepackage{float}	
\usepackage{geometry} 
\usepackage{fancyhdr} 
\usepackage{longtable}
\usepackage{listings}
\usepackage{subfigure}
\usepackage{hyperref} 
\usepackage[usenames,dvipsnames]{color} 
\usepackage{lineno}
\usepackage{enumitem}
\usepackage{natbib}
\setcitestyle{numbers,square}

\theoremstyle{plain}
\newtheorem{theorem}{Theorem}[section]
\newtheorem{lemma}[theorem]{Lemma}

\newtheorem{Proposition}[theorem]{Proposition}

\theoremstyle{definition}
\newtheorem{definition}[theorem]{Definition}

\theoremstyle{remark}

\newcommand{\R}{{\mathbb R}}

\numberwithin{equation}{section}
\newlength{\mahaut}%
\newlength{\malarg}%

\newcommand{\intevariable}[3]{
	\settowidth{\malarg}{$\displaystyle \int $}
	\settoheight{\mahaut}{$\displaystyle #3 $}%
	\setlength{\mahaut}{1.2\mahaut}
	{\resizebox{\malarg}{\mahaut}{$\displaystyle\int$}}_{\kern-0.5\malarg#1}^{#2}%
	#3%
}

 \definecolor{violet}{rgb}{0.6, 0.4, 0.8}
 \definecolor{violetb}{rgb}{0.7, 0.2, 0.5}

\makeatletter
\def\@setcopyright{}
\def\serieslogo@{}
\makeatother

\begin{document}
	
	\author{Véronique Maume-Deschamps}
	\address{Université de Lyon, Université Lyon 1, France, Institut Camille Jordan UMR 5208}
	\email{veronique.maume@univ-lyon1.fr}

	\author{Didier Rulli\`ere}
	\address{Université de Lyon, Université Lyon 1, France, Laboratoire SAF EA 2429}
	\email{didier.rulliere@univ-lyon1.fr}

	 \author{Khalil Said}
	 \address{\'Ecole d'Actuariat, Université Laval, Québec, Canada}
	 \email{khalil.said.1@ulaval.ca}
	

	 \title[Extremes for multivariate expectiles]{Extremes for multivariate expectiles}

	
  \begin{abstract}
  	In \cite{maumedeschamps3}, a new family of vector-valued risk measures called \textit{multivariate expectiles} is introduced. In this paper, we focus on the asymptotic behavior of these measures in a multivariate regular variations context. For models with equivalent tails, we propose an estimator of extreme multivariate expectiles, in the Fréchet attraction domain case, with asymptotic independence, or for comonotonic margins.   
  \end{abstract}										
	\date{\today}
	\maketitle
	
\section*{Introduction}
In few years, expectiles became an important risk measure among more used ones, essentially, because it satisfies both coherence and elicitability properties.
In dimension one, expectiles were introduced by Newey and Powell (1987)~\cite{Expect1987}. For a random variable $X$ with finite order $2$ moment, the expectile of level $\alpha$ is defined as 
\begin{equation*}\label{ExpDef}
e_\alpha(X)=\arg\min_{x\in \mathbb{R}}\mathbb{E}[\alpha(X-x)_+^{2}+(1-\alpha)(x-X)_+^{2}]\/, 
\end{equation*}
where  $(x)_+=\max(x,0)$. Expectiles are the only risk measure satisfying both elicitability and coherence properties, according to Bellini and Bignozzi (2015) \cite{bellini2015elicitable}. \\ 

In higher dimension, one of the proposed extensions of expectiles in \cite{maumedeschamps3} are \textit{Matrix Expectiles}. Consider a random vector $\mathbf{X}=(X_1,\ldots,X_d)^T\in\mathbb{R}^d$ having order $2$ moments, and let $\Sigma=(\pi_{ij})_{1\leq i,j\leq d}$ be a $d\times d$ real matrix, symmetric and positive semi-definite such that $i\in\{1,\ldots,d\}, ~\pi_{ii}=\pi_{i}>0$.\\
A $\Sigma$-expectile of $\mathbf{X}$, is defined as
\begin{equation*}
\mathbf{e}^\Sigma_\alpha(\mathbf{X})\in\underset{\mathbf{x}\in\mathbb{R}^d}{\arg\min}~\mathbb{E}[\alpha(\mathbf{X}-\mathbf{x})^T_+\Sigma(\mathbf{X}-\mathbf{x})_++(1-\alpha)(\mathbf{X}-\mathbf{x})^T_-\Sigma(\mathbf{X}-\mathbf{x})_-]\/,	\end{equation*}
where $(\mathbf{x})_+=((x_1)_+,\ldots,(x_d)_+)^T$ and $(\mathbf{x})_-=(\mathbf{-x})_+$. 
We shall concentrate on the case where the above minimization has a unique solution. In \cite{maumedeschamps3},  conditions on $\Sigma$ ensuring the uniqueness of the argmin are given, it is sufficient that $ \pi_{ij}\geq0,~\forall i,j\in\{1,\ldots,d\}$. We shall make this assumption throughout this paper. Then, the vector expectile is unique, and it is solution of the following equations system 
\begin{equation}
\alpha\sum_{i=1}^{d}\pi_{ki}\mathbb{E}[(X_i-x_i)_+ 1\!\!1_{\{X_k>x_k\}}]
=(1-\alpha)\sum_{i=1}^{d}\pi_{ki}\mathbb{E}[(x_i-X_i)_+ 1\!\!1_{\{x_k>X_k\}}],~~\forall k\in\{1,\ldots,d\}\/.
\label{eq1}
\end{equation}
 In case $\pi_{ij}=1$ for all $(i,j)\in\{1,\ldots,d\}^2$, the corresponding $\Sigma$-expectile is called a \textit{$L_1$-expectile}. It coincides with the $L_1$-norm expectile defined in \cite{maumedeschamps3}.\\
In \cite{maumedeschamps3} it is proved that, 

$$\underset{\alpha\longrightarrow1}{\lim}\mathbf{e}^\Sigma_\alpha(\mathbf{X})=\mathbf{X_F},~~\mbox{and}~~\underset{\alpha\longrightarrow0}{\lim}\mathbf{e}^\Sigma_\alpha(\mathbf{X})=\mathbf{X_I}
\/,$$
where $\mathbf{X_F}\in (\R\cup \{+\infty\})^d$ is the right endpoint vector $(x^1_F,\ldots,x^d_F)^T$, and by $\mathbf{X_I}\in (\R\cup \{-\infty\})^d$ is the left endpoint vector $(x^1_I,\ldots,x^d_I)^T$ of the support of the random vector $\mathbf{X}$.\\

The multivariate expectiles can be estimated in the general case using stochastic optimization algorithms. The example of estimation by the Robbins-Monro's (1951) \cite{robbins1951stochastic} algorithm, presented in \cite{maumedeschamps3}, shows that for extreme levels, the obtained estimation is not satisfactory in term of convergence speed. This leads us to the theoretical analysis of the asymptotic behavior of multivariate expectiles. Asymptotic levels i.e. $\alpha \rightarrow 1$ or $\alpha \rightarrow 0$ represent  extreme risks. Since the solvency thresholds in insurance are generally high (e.g. $\alpha=0.995$ for Solvency II directive), the study of asymptotic behavior of risk measures is of natural importance. The goal of this work is to establish the asymptotic behaviour of multivariate expectiles. The study of the extreme behaviour of risk measures in a multivariate regular variation framework is the subject of a balk of works, let us mention as examples, Embrechts et al. (2009) \cite{RV1}, Albrecher et al. (2006) \cite{RV5} in risk aggregation contexts, and Asimit et al. (2011) \cite{RV4} for risk capital allocation. Similar works are also done on other multivariate risk measures, as example, for the Multivariate Conditional-Tail-Expectation in a recent paper of Di Bernardino and Prieur \cite{RV2}.\\

 We shall work on the equivalent tails model. It is often used in modeling the claim amounts in insurance, in studying dependent extreme events, and in ruin theory models. This model includes in particular the identically distributed portfolios of risks and the case with scale difference in the distributions. In this paper, we study the asymptotic behavior of multivariate expectiles in the multivariate regular variations framework. We focus on marginal distributions  belonging  to the Fréchet domain of attraction. This domain contains heavy-tailed distributions that represent the most dangerous claims in insurance. Let us remark that the attention to univariate expectiles is recent. In \cite{belliniElena}, asymptotic equivalents of expectiles as a function of the quantile of the same level for regular variation distributions are proved. First and second order asymptotics for the expectile of the sum in the case of FGM dependence structure are given in  \cite{ExpectileFGM}.\\ 
 \ \\

The paper is constructed as follows. The first section is devoted to the presentation of the multivariate regularly varying distribution framework.  The study of the asymptotic behavior of the multivariate expectiles for Fréchet model with equivalent tails is the subject of Section 2. The case of an asymptotically dominant tail is analyzed in Section 3. Section 4 is devoted to estimations of extreme multivariate expectiles in the cases of asymptotic independence and comonotonicity. Numerical illustrations are given using simulations in different models.


\section{The MRV Framework}
Regularly varying distributions are well suited to study extreme phenomenons. Lots of works have been devoted to the asymptotic behavior of usual risk measures for this class of distributions, and  results are given for  sums of risks  belonging to this family.  It is well known that the three domains of attraction of extreme value distributions can be defined  using the concept of regular variations (see \cite{embrechts1997,resnick2007,de2007extreme,bingham1989regular}). \\

This section is devoted to the classical characterization of multivariate regular variations, which will be used in the study of the asymptotic behavior of multivariate expectiles. We also recall some basic results on the univariate setting that we shall use.
\subsection{Univariate regular variations}
We begin by recalling basic definitions and results on univariate regular variations.
	\begin{definition}[Regularly varying functions]
		A measurable positive function $f$ is \textit{regularly varying of index $\rho$ at} $a\in\{0,+\infty\}$, if for all  $t>0$, 
		$$\underset{x\rightarrow a}{\lim}\frac{f(tx)}{f(x)}=t^\rho\/,$$
		we denote $f\in\mbox{RV}_\rho(a)$.
	\end{definition}
	A slowly varying function is a regularly varying function of index $\rho=0$. Remark that $f\in\mbox{RV}_\rho(+\infty)$ if and only if, there exists a slowly varying function at infinity, $L\in\mbox{RV}_0(+\infty)$ such that 
	$$f(x)=x^\rho L(x)\/.$$

\begin{theorem}[Karamata's representation,  \cite{resnick2013extreme}]\label{Karamata-Rep} For any slowly varying function $L$ at $+\infty$, there exist a positive measurable function $c(\cdot)$ that satisfies  $\underset{x\rightarrow+\infty}{\lim}c(x)=c\in]0,+\infty[$, and a measurable function $\varepsilon(\cdot)$ with  $\underset{x\rightarrow+\infty}{\lim}\varepsilon(x)=0$, such that 
	$$L(x)=c(x)\exp\left(\int_{1}^{x}\frac{\varepsilon(t)}{t}dt\right)\/.$$ 	
\end{theorem}
The Karamata's representation is generalized to  RV functions. Indeed, $f\in\mbox{RV}_\rho(+\infty)$ if and only if it can written in the form  
$$f(x)=c(x)\int_{1}^{x}\frac{\rho(t)}{t}dt\/,$$
where  $\underset{t\rightarrow\infty}{\rho(t)}=\rho$ and $\underset{t\rightarrow\infty}{c(t)}=c\in]0,+\infty[$.\\
Throughout the paper, we shall consider generalized inverses of non-decreasing functions $f$: $f^{\longleftarrow}(y)=\inf\{x\in\R \/, \ f(x)\geq y \}$.
\begin{lemma}[Inverse of RV functions \cite{resnick2007}]\label{inevrseRV}
Let $f$ be a measurable non-decreasing function defined on $\mathbb{R}^+$, such that $\underset{x\rightarrow+\infty}{\lim}f(x)=+\infty$. Then 
$$f\in\mbox{RV}_\rho(+\infty)~~\mbox{if and only if}~~f^{\longleftarrow}\in\mbox{RV}_{\frac{1}{\rho}}(+\infty)\/,$$
for all $0\leq\rho\leq+\infty$, where we follow the convention $1/0=\infty$ and $1/\infty=0$.
\end{lemma}
\begin{lemma}[Integration of RV functions (Karamata's Theorem)), \cite{mikosch2003modeling}]\label{Th-Karamata} For a positive measurable function $f$, regularly varying  of index $\rho$ at $+\infty$, locally bounded on $[x_0, +\infty)$ with $x_0\geq 0$
	\begin{itemize}
		\item if $\rho>-1$, then  $$\underset{x\rightarrow+\infty}{\lim}\displaystyle\frac{\displaystyle\int_{x_0}^{x}f(t)dt}{xf(x)}=\frac{1}{\rho+1}\/,$$
		\item if $\rho<-1$, then $$\underset{x\rightarrow+\infty}{\lim}\displaystyle\frac{\displaystyle\int_{x}^{+\infty}f(t)dt}{xf(x)}=-\frac{1}{\rho+1}\/.$$
	\end{itemize}
	
\end{lemma}	
\begin{lemma}[Potter's bounds \cite{bingham1989regular}]\label{Potter}
	For $f\in\mbox{RV}_\rho(a)$, with $a\in\{0,\infty\}$ and $\rho\in\mathbb{R}$. For any $0<\epsilon<1$ and all $x$ and $y$ sufficiently close to $a$, we have
	$$(1-\epsilon)\min\left(\left(\frac{x}{y}\right)^{\rho-\epsilon},\left(\frac{x}{y}\right)^{\rho+\epsilon}\right)\leq\frac{f(x)}{f(y)}\leq(1+\epsilon)\max\left(\left(\frac{x}{y}\right)^{\rho-\epsilon},\left(\frac{x}{y}\right)^{\rho+\epsilon}\right)\/.$$ 
\end{lemma}
Many other properties of regularly varying functions are presented e.g. in \cite{bingham1989regular}. 

  \subsection{Multivariate regular variations}
  The multivariate extension of regular variations is introduced in \cite{MRVintro}. We denote by $\mu_n\overset{v}{\longrightarrow}\mu$ the vague convergence of Radon measures as presented in \cite{kallenbergVagueC}. The following definitions are given for non negative random variables.
  \begin{definition}[Multivariate regular variations]\label{MRVd1}
  	The distribution of a random vector $\mathbf{X}$ on $[0,\infty]^d$ is said to be regularly varying if there exist a non-null Radon measure $\mu_\mathbf{X}$ on the Borel  $\sigma$-algebra $\mathcal{B}_d$ on $[0,\infty]^d\backslash\mathbf{0}$, and a normalization function $b: \mathbb{R}\longrightarrow\mathbb{R}$ which satisfies  $\underset{x\longrightarrow+\infty}{\lim}b(x)=+\infty$ such that
  	\begin{equation}\label{DefMRV-eq1}
  	u\mathbb{P}\left(\frac{\mathbf{X}}{b(u)}\in \cdot\right)\overset{\upsilon}{\longrightarrow}\mu_\mathbf{X}(\cdot) \mbox{ as }u\longrightarrow+\infty\/.
  	\end{equation} 
  \end{definition}
  There exist several equivalent definitions of multivariate regular variations which will be useful in what follows. 
  \begin{definition}[MRV equivalent definitions]
  	Let $\mathbf{X}$ be a random vector on $\mathbb{R}^d$, the following definitions are equivalent:
  	\begin{itemize}
  		\item The vector $\mathbf{X}$ has a regularly varying tail of index $\theta$.  
  		\item There exist a finite measure $\mu$ on the unit sphere $\mathbb{S}^{d-1}$, and a normalization function  $b:(0,\infty)\longrightarrow(0,\infty)$ such that
  		\begin{equation}\label{DefMRV-eq2}
  		\underset{t\longrightarrow+\infty}{\lim}\mathbb{P}\left(\lVert\mathbf{X}\lVert>xb(t),\frac{\mathbf{X}}{\lVert\mathbf{X}\lVert}\in .\right)=x^{-\theta}\mu(.)\/,
  		\end{equation}
  		for all $x>0$. The measure $\mu$ depends on the chosen norm, it is called the \textit{spectral measure} of $\mathbf{X}$. 
  		\item There exist a finite measure $\mu$ on the unit sphere $\mathbb{S}^{d-1}$, a slowly varying function $L$, and a positive real $\theta>0$ such that
  		\begin{equation}\label{DefMRV-eq3}
  		\underset{x\longrightarrow+\infty}{\lim}\frac{x^{\theta}}{L(x)}\mathbb{P}\left(\lVert\mathbf{X}\lVert>x,\frac{\mathbf{X}}{\lVert\mathbf{X}\lVert}\in B\right)=\mu(B)\/,
  		\end{equation}   
  		for all $B\in\mathcal{B}(\mathbb{S}^{d-1})$ with $\mu(\partial B)=0$.
  	\end{itemize}
  \end{definition}
  From now on, MRV denotes the set of multivariate regularly varying distributions, and MRV$(\theta,\mu)$ denotes the set of random vectors with regularly varying tail, with index $\theta$ and spectral measure $\mu$. 
  \\
  From (\ref{DefMRV-eq3}), we may assume that $\mu$ is normalized i.e. $\mu(\mathbb{S}^{d-1})=1$, which implies that  $\lVert\mathbf{X}\lVert$  has a regularly varying tail of index $-\theta$.\\
 On another hand,
  \begin{align*}
  \underset{x\longrightarrow+\infty}{\lim}\mathbb{P}\left(\frac{\mathbf{X}}{\lVert\mathbf{X}\lVert}\in B \biggm\lvert \lVert\mathbf{X}\lVert>x,\right)&=\underset{x\longrightarrow+\infty}{\lim}\frac{\mathbb{P}\left(\lVert\mathbf{X}\lVert>x,\frac{\mathbf{X}}{\lVert\mathbf{X}\lVert}\in B\right)}{\mathbb{P}\left(\lVert\mathbf{X}\lVert>x\right)}\\
  &=\underset{x\longrightarrow+\infty}{\lim}\frac{x^{\theta}}{L(x)}\mu(B)x^{-\theta}L(x)=\mu(B)\/,
  \end{align*}
  for all $B\in\mathcal{B}(\mathbb{S}^{d-1})$ with $\mu(\partial B)=0$. That means that conditionally to  $\{\lVert\mathbf{X}\lVert>x\}$,   $\frac{\mathbf{X}}{\lVert\mathbf{X}\lVert}$ converges weakly to $\mu$.\\
  The different possible characterizations of the MRV concept are presented in \cite{mikosch2003modeling}.
  \subsection{Characterization using tail dependence functions}
  Let $\mathbf{X}=(X_1,\ldots,X_d)$ be a random vector. From now on, $\overline{F}_{X_i}$ denotes the survival function of $X_i$.   In this paper, we use the definition of the upper tail dependence function, as introduced in \cite{kluppelberg2008semi}.
  \begin{definition}[The tail dependence function]
  	Let $\mathbf{X}$ be a random vector on $\mathbb{R}^d$, with continuous marginal distributions. The tail dependence function is defined by
  	\begin{equation}\label{LU}
  	\lambda_U^{\mathbf{X}}(x_1,\ldots,x_d)=\underset{t\longrightarrow 0}{\lim}t^{-1}\mathbb{P}(\bar{F}_{X_1}(X_1)\leq tx_1,\ldots,\bar{F}_{X_d}(X_d)\leq tx_d)\/,
  	\end{equation} 
  	when the limit exists.  
  \end{definition}  
  For $k\leq d$, denote by $X^{(k)}$ a $k$ dimensional sub-vector of $\mathbf{X}$, $C^{(k)}$ its copula and  $\overline{C}^{(k)}$ its survival copula. The upper tail dependence function is
  \begin{equation}\label{LUC}
  \lambda_U^k(u_1,\ldots,u_k)=\underset{t\longrightarrow0^+}{\lim}\frac{\bar{C}^{(k)}(tu_1,\ldots,tu_k)}{t}\/,
  \end{equation}
  if this limit exists. 
  The lower tail dependence function can be defined analogically by
  $$\lambda_L^k(u_1,\ldots,u_k)=\underset{t\longrightarrow0^+}{\lim}\frac{C^{(k)}(tu_1,\ldots,tu_k)}{t},$$
  when the limit exists. In this paper, our study is limited to the upper version as defined in (\ref{LUC}).\\
  
  We assume that $\mathbf{X}$ has equivalent regularly varying marginal tails, which means:
  \begin{description}
  	\item[H1] $\bar{F}_{X_1}\in\mbox{RV}_{-\theta}(+\infty)\/,$  	with $\theta>0$.
  	\item[H2] The tails of  $X_i,i=1,\ldots,d$ are equivalent. That is for all $i\in\{2,\ldots,d\}$, there is a  positive constant $c_i$ such that
  	$$\underset{x\longrightarrow+\infty}{\lim}\frac{\bar{F}_{X_i}(x)}{\bar{F}_{X_1}(x)}=c_i\/.$$
  \end{description} 
  H1 and H2 imply that all marginal tails are regularly varying of index $-\theta$ at $+\infty$.\\ 

  The following two theorems show that, under H1 and H2, the MRV character of multivariate distributions is equivalent to the existence of the tail dependence functions.
  \begin{theorem}[Theorem 2.3 in  \cite{li2009tail}]\label{Th-MRV1}
  	Let $\mathbf{X}=(X_1,\ldots,X_d)$ be a random vector in $\mathbb{R}^d$, with continuous marginal distributions $F_{X_i}, i=1,\ldots,d$ that satisfy H1 and H2. If $\mathbf{X}$ has a MRV distribution, the tail dependence function exists, and it is given by par$$\lambda^k_U(u_1,\ldots,u_k)=\underset{x\longrightarrow+\infty}{\lim}x\mathbb{P}\left(X_1>b(x)\left(\frac{u_1}{c_1}\right)^{-1/\theta},\ldots,X_k>b(x)\left(\frac{u_d}{c_d}\right)^{-1/\theta}\right)\/,$$ 
  	for any $k\in\{1,\ldots,d\}$. 
  \end{theorem}
  \begin{theorem}[Theorem 3.2 in \cite{MRVcopules}]\label{MRVcopules}
  	Let $\mathbf{X}=(X_1,\ldots,X_d)$ be a random vector in $\mathbb{R}^d$, with continuous marginal distributions $F_{X_i}, i=1,\ldots,d$ that satisfies H1 and H2. If the tail dependence function $\lambda^k_U$ exists for all $k\in\{1,\ldots,d\}$, then $\mathbf{X}$ is MRV, its normalization function is given by  $b(u)=\left(\frac{1}{\bar{F}_{X_1}}\right)^{\longleftarrow}(u)$ and the spectral measure is  $$\mu([\mathbf{0},\mathbf{x}]^c)=\sum_{i=1}^{d}c_ix_i^{-\theta}-\sum_{1\leq i<j\leq d}^{}\lambda_U^2(c_ix_i^{-\theta},c_jx_j^{-\theta})+\cdots+(-1)^{d+1}\lambda_U^d(c_1x_1^{-\theta},\ldots,c_dx_d^{-\theta})\/.$$
  \end{theorem} 
  By construction of the multivariate expectiles, only the bivariate dependence structures  are taken into account. We shall use the functions $\lambda^{(X_i,X_k)}_U$, for all $(i,k)\in\{1,\ldots,d\}^2$. In order to simplify the notation, we denote it by $\lambda^{ik}_U$. If the vector $\mathbf{X}$ has an MRV distribution, the pairs $ (X_i, X_j) $ have also  MRV distributions, for any $(i,j)\in\{1,\ldots,d\}^2$. So, in the MRV framework, and under H1 and H2, the existence of functions $\lambda^{ik}$ is insured. In addition, we assume in all the rest of this paper that these functions are continuous.
 \section{Fréchet model with equivalent tails}   \label{sec:equiv}
 In this section, we assume that $\mathbf{X}$ satisfies H1 and H2 with $\theta>1$. It implies that $X_1$ belongs to the extreme value domain of attraction of Fréchet $\mbox{MDA}(\Phi_\theta)$. This domain contains distributions with infinite endpoint $x_F=\sup\{x:F(x)<1\}=+\infty$, so as $\alpha\longrightarrow 1$ we get $\mathbf{e}^i_\alpha(\mathbf{X}) \longrightarrow +\infty$ $\forall i$. Also, from Karamata's Theorem (Theorem \ref{Th-Karamata}), we have for $i=1\/,\ldots\/,d$,
 \begin{equation}\label{EqFrechetSL}
 \underset{x \longrightarrow +\infty }{\lim}\frac{\mathbb{E}[(X_i-x)_+]}{x\bar{F}_{X_i}(x)}=\frac{1}{\theta-1}\/,
 \end{equation}
 for all $i\in\{1,\ldots,d\}$. 
 \begin{Proposition}\label{PropLimits1}
 	Let $\Sigma=(\pi_{ij})_{i\/,j=1\/,\ldots\/,d}$ with $\pi_{ij}>0$ for all $i\/,j\in\{1\/,\ldots\/,d\}$. Under H1 and H2, the components of the multivariate $\Sigma$-expectiles $\mathbf{e}_\alpha(\mathbf{X})=(\mathbf{e}^i_\alpha(\mathbf{X}))_{i=1\/,\ldots\/,d}$ satisfy
 	$$
 	0<\underset{\alpha\longrightarrow 1}{\underline{\lim}}\frac{\mathbf{e}^i_\alpha(\mathbf{X})}{\mathbf{e}^1_\alpha(\mathbf{X})}\leq\underset{\alpha\longrightarrow 1}{\overline{\lim}}\frac{\mathbf{e}^i_\alpha(\mathbf{X})}{\mathbf{e}^1_\alpha(\mathbf{X})}<+\infty, \forall i\in\{2,\ldots,d\}\/.
 	$$ 	
 \end{Proposition}
 Proposition \ref{PropLimits1} implies that distributions with equivalent tails have asymptotically comparable multivariate expectile components.\\
 
 Before we prove Proposition \ref{PropLimits1}, we shall demonstrate some preliminary results. Firstly, let  $\mathbf{X}=(X_1,\ldots,X_d)^T$  satisfy H1 and H2, we denote $x_i=\mathbf{e}^i_\alpha(\mathbf{X})$ for all $i\in\{1,\dots,d\}$. We define the functions $l^\alpha_{X_i,X_j}$ for all $(i,j)\in\{1,\ldots,d\}^2$ by
 \begin{equation}\label{DefFL}
 l^\alpha_{X_i,X_j}(x_i,x_j)=\alpha\mathbb{E}[(X_i-x_i)_+ 1\!\!1_{\{X_j>x_j\}}]-(1-\alpha)\mathbb{E}[(X_i-x_i)_- 1\!\!1_{\{X_j<x_j\}}]\/,
 \end{equation}
 and $l^\alpha_{X_i}(x_i)=l^\alpha_{X_i,X_i}(x_i,x_i)$.\\
 The optimality system (\ref{eq1}) rewrites 
 \begin{equation}\label{S-functionsL2}
 l^\alpha_{X_k}(x_k)=-\sum_{i=1,i\neq k}^{d}\frac{\pi_{ki}}{\pi_{kk}}l^\alpha_{X_i,X_k}(x_i,x_k)~~\forall k\in\{1,\ldots,d\}\/.
 \end{equation}
 We shall use the following sets:
 $$J^i_0=\{j\in\{1,\ldots,d\}\setminus\{i\}\mid \underset{\alpha\longrightarrow1}{\varliminf}\frac{x_j}{x_i}=0\}\/,$$
 $$J^i_C=\{j\in\{1,\ldots,d\}\setminus\{i\}\mid  0<\underset{\alpha\longrightarrow1}{\varliminf}\frac{x_j}{x_i} <\underset{\alpha\longrightarrow1}{\varlimsup}\frac{x_j}{x_i} <+\infty\}\/,$$
 $$\mbox{and }~J^i_\infty=\{j\in\{1,\ldots,d\}\setminus\{i\}\mid \underset{\alpha\longrightarrow1}{\varlimsup}\frac{x_j}{x_i}=+\infty\}\/.$$
 The proof of Proposition \ref{PropLimits1} is written for  $\pi_{ij}=1$, for all $(i,j)\in\{1,\ldots,d\}^2$, ie for the $L_1$-expectiles. The general case can be treated in the same way, provided that $\pi_{ij}>0$ for all $(i,j)\in\{1,\ldots,d\}^2$. The proof of Proposition \ref{PropLimits1} follows from Lemma \ref{Lemme+inf} and Proposition \ref{PropLimits1V1} below.	
 \begin{lemma}\label{Lemme+inf}
 	Assume that H1 and H2 are satisfied. 
 		\begin{enumerate}
 			\item If $t\overset{}{=}o(s)$ then for all $(i,j)\in\{1,\ldots,d\}^2$,
 			$$\underset{t\rightarrow+\infty}{\lim}\frac{s\bar{F}_{X_i}(s)}{t\bar{F}_{X_j}(t)}=0\/.$$
 			\item If $t=\Theta(s)$,\footnote{Recall that $t=\Theta(s)$ means that there exist positive constants $C_1$ and $C_2$ such that $C_1 s \leq t \leq C_2 s$} then for all $(i,j)\in\{1,\ldots,d\}^2$,
 			$$\frac{\overline{F}_{X_i}(s)}{\overline{F}_{X_j}(t)}\sim \frac{c_i}{c_j}\left(\frac{s}{t}\right)^{-\theta}\ \mbox{as} \ t\rightarrow \infty\/.$$
 		\end{enumerate} 	
 \end{lemma}
The proof is given in Appendix \ref{pro1}.
 \begin{Proposition}\label{PropLimits1V1}
 	Under H1 and H2, the components of the extreme multivariate expectile satisfy
  $$
 0<\underset{\alpha\longrightarrow 1}{\underline{\lim}}\frac{1-\alpha}{\bar{F}_{X_i}(\mathbf{e}^i_\alpha(\mathbf{X}))}\leq\underset{\alpha\longrightarrow 1}{\overline{\lim}}\frac{1-\alpha}{\bar{F}_{X_i}(\mathbf{e}^i_\alpha(\mathbf{X}))}<+\infty, \forall i\in\{2,\ldots,d\}\/.
 $$ 	
 \end{Proposition}
 The proof is given in Appendix \ref{pro2}.\\
 We may now prove Proposition \ref{PropLimits1}. 
 \begin{proof}[Proof of Proposition \ref{PropLimits1}] 
 	We shall prove that $J_\infty^1=\emptyset$, the fact that $J_\infty^k=\emptyset$ for all $k\in\{1\/,\ldots\/,d\}$ may be proven in the same way. This implies that $J_0^k=J_\infty^k=\emptyset$ for all $k\in\{1,\ldots,d\}$, hence the result.\\
 	We suppose that $J_\infty^1\neq\emptyset$, let $i\in J_\infty^1$, taking if necessary a subsequence, we may assume that $x_i/x_1\rightarrow +\infty$ as $\alpha\rightarrow 1$.\\
 	From Proposition \ref{PropLimits1V1}, we have
 	$$
 	0<\underset{\alpha\longrightarrow 1}{\underline{\lim}}\frac{1-\alpha}{\bar{F}_{X_i}(\mathbf{e}^i_\alpha(\mathbf{X}))}\leq\underset{\alpha\longrightarrow 1}{\overline{\lim}}\frac{1-\alpha}{\bar{F}_{X_i}(\mathbf{e}^i_\alpha(\mathbf{X}))}<+\infty, \forall i\in\{2,\ldots,d\}\/,
 	$$ 
 	so, taking if necessary a  subsequence, we may assume that $\exists\ell\in\mathbb{R}^*\backslash\{+\infty\}$ such that $$\underset{\alpha\longrightarrow1}{\lim}\frac{1-\alpha}{\bar{F}_{X_1}(x_1)}=\ell\/.$$
 	In this case,  $$\underset{\alpha\longrightarrow 1}{\lim}\frac{l^\alpha_{X_1}(x_1)}{x_1\bar{F}_{X_1}(x_1)}=\underset{\alpha\longrightarrow 1}{\lim}\left((2\alpha-1)\frac{\mathbb{E}[(X_1-x_1)_+]}{x_1\bar{F}_{X_1}(x_1)}-\frac{1-\alpha}{\bar{F}_1(x_1)}(1-\frac{\mathbb{E}[X_1]}{x_1})\right)=\frac{1}{\theta-1}-\ell<+\infty\/.$$
 	Moreover, 
 	\begin{align*}
 	\frac{\mathbb{E}[(X_i-x_i)_+ 1\!\!1_{\{X_1>x_1\}}]}{x_1\bar{F}_{X_1}(x_1)}&\leq\frac{\mathbb{E}[(X_i-x_i)_+ ]}{x_1\bar{F}_{X_1}(x_1)}=\frac{\mathbb{E}[(X_i-x_i)_+]}{x_i\bar{F}_{X_i}(x_i)}\frac{x_i\bar{F}_{X_i}(x_i)}{x_1\bar{F}_{X_1}(x_1)}\longrightarrow 0 \ \mbox{using Lemma \ref{Lemme+inf}\/.}
 	\end{align*}
 	We get
 	\begin{align*}
 	\underset{\alpha\longrightarrow 1}{\lim}\frac{l^\alpha_{X_i,X_1}(x_i,x_1)}{x_1\bar{F}_{X_1}(x_1)}&=\underset{\alpha\longrightarrow 1}{\lim}\left(\frac{\alpha\mathbb{E}[(X_i-x_i)_+ 1\!\!1_{\{X_1>x_1\}}]-(1-\alpha)\mathbb{E}[(X_i-x_i)_- 1\!\!1_{\{X_1<x_1\}}]}{x_1\bar{F}_{X_1}(x_1)}\right)\\
 	&=\underset{\alpha\longrightarrow 1}{\lim}\left(\frac{\mathbb{E}[(X_i-x_i)_+ 1\!\!1_{\{X_1>x_1\}}]}{x_1\bar{F}_{X_1}(x_1)}-\frac{1-\alpha}{\bar{F}_{X_1}(x_1)}\frac{x_i}{x_1}\right)=-\infty, ~\forall i\in J_\infty^1\/.
 	\end{align*}
 	
 	Going through the limit ($\alpha\longrightarrow1$) in the first equation of the optimality System (\ref{S-functionsL2}) divided by $x_1\bar{F}_{X_1}(x_1)$, leads to 
 	\begin{equation}\label{Case2-Eq1}
 	\underset{\alpha\longrightarrow 1}{\lim}\sum_{k\in J^1_0\cup J^1_C\setminus J_\infty^1}^{}\frac{l^\alpha_{X_k,X_1}(x_k,x_1)}{x_1\bar{F}_{X_1}(x_1)}=-\infty\/.
 	\end{equation}
 	Now, let $k\in J^1_0$
 	\begin{align*}
 	\frac{\mathbb{E}[(X_k-x_k)_+ 1\!\!1_{\{X_1>x_1\}}]}{x_1\bar{F}_{X_1}(x_1)}&=\frac{\displaystyle\int_{x_k}^{x_1}\mathbb{P}\left(X_k>t,X_1>x_1\right)dt}{x_1\bar{F}_{X_1}(x_1)}+\frac{\displaystyle\int_{x_1}^{+\infty}\mathbb{P}\left(X_k>t,X_1>x_1\right)dt}{x_1\bar{F}_{X_1}(x_1)}\\
 	&\leq\frac{\displaystyle\int_{x_k}^{x_1}\mathbb{P}\left(X_1>x_1\right)dt}{x_1\bar{F}_{X_1}(x_1)}+\frac{\displaystyle\int_{x_1}^{+\infty}\mathbb{P}\left(X_k>t\right)dt}{x_1\bar{F}_{X_1}(x_1)}\/,
 	\end{align*}
 	Karamata's Theorem (Theorem \ref{Th-Karamata}) leads to
 	$$\underset{\alpha\longrightarrow 1}{\lim}\frac{\displaystyle\int_{x_k}^{x_1}\mathbb{P}\left(X_1>x_1\right)dt}{x_1\bar{F}_{X_1}(x_1)}+\frac{\displaystyle\int_{x_1}^{+\infty}\mathbb{P}\left(X_k>t\right)dt}{x_1\bar{F}_{X_1}(x_1)}=1+\frac{c_k}{\theta-1}, \forall k\in J^1_0\/.$$
 	Consider  $k\in J^1_C$
 	\begin{align*}
 	\frac{\mathbb{E}[(X_k-x_k)_+ 1\!\!1_{\{X_1>x_1\}}]}{x_1\bar{F}_{X_1}(x_1)}&\leq\frac{\mathbb{E}[(X_k-x_k)_+ ]}{x_1\bar{F}_{X_1}(x_1)}=\frac{\mathbb{E}[(X_k-x_k)_+]}{x_k\bar{F}_{X_k}(x_k)}\frac{x_i\bar{F}_{X_k}(x_k)}{x_1\bar{F}_{X_1}(x_1)}\/,
 	\end{align*}
 	and
 	$$\frac{\mathbb{E}[(X_k-x_k)_+]}{x_i\bar{F}_{X_k}(x_k)}\frac{x_k\bar{F}_{X_k}(x_k)}{x_1\bar{F}_{X_1}(x_1)}\underset{\alpha\longrightarrow 1}{\sim}\frac{c_k}{\theta-1}\left(\frac{x_k}{x_1}\right)^{-\theta+1}\/.$$
 	Finally, we deduce that
 	\begin{align*}
 	-\sum_{k\in J^1_0\cup J^1_C\setminus J_\infty^1}^{}\underset{\alpha\longrightarrow 1}{\varlimsup}\frac{x_k}{x_1}&\leq\underset{\alpha\longrightarrow 1}{\varliminf}\sum_{k\in J^1_0\cup J^1_C\setminus J_\infty^1}^{}\frac{l^\alpha_{X_k,X_1}(x_k,x_1)}{x_1\bar{F}_{X_1}(x_1)}\\&\leq\underset{\alpha\longrightarrow 1}{\varlimsup}\sum_{k\in J^1_0\cup J^1_C\setminus J_\infty^1}^{}\frac{l^\alpha_{X_k,X_1}(x_k,x_1)}{x_1\bar{F}_{X_1}(x_1)}\\
 	&\leq\sum_{k\in  J^1_C}\left(\frac{c_k}{\theta-1}\left(\underset{\alpha\longrightarrow 1}{\varlimsup}\frac{x_k}{x_1}\right)^{-\theta+1}-\ell\underset{\alpha\longrightarrow 1}{\varliminf}\frac{x_k}{x_1}\right)+\sum_{k\in  J^1_0\setminus J_\infty^1}\left(1+\frac{c_k}{\theta-1}\right)\/.
 	\end{align*}
 	This is contradictory with  (\ref{Case2-Eq1}), and consequently $J^1_\infty$ is necessarily an empty set. The result follows.
 \end{proof}
 \begin{Proposition}[Extreme multivariate expectile]\label{PropExpEx}
Assume  that H1 and H2 are satisfied and  $\mathbf{X}$ has a regularly varying multivariate distribution in the sense of Definition \ref{MRVd1}. Consider the $L_1$-expectiles $\mathbf{e}_\alpha(\mathbf{X})= (\mathbf{e}_\alpha^i(\mathbf{X}))_{i=1\/,\ldots\/,d}$. Then any limit vector $(\eta,\beta_2,\ldots,\beta_d)$ of $\left(\frac{1-\alpha}{\bar{F}_{X_1}(\mathbf{e}_\alpha^1(\mathbf{X}))},\frac{\mathbf{e}_\alpha^2(\mathbf{X})}{\mathbf{e}_\alpha^1(\mathbf{X})},\ldots,\frac{\mathbf{e}_\alpha^d(\mathbf{X})}{\mathbf{e}_\alpha^1(\mathbf{X})}\right)$ satisfies the following equation system
 	\begin{equation}\label{EqInt1}
 	\frac{1}{\theta-1}-\eta\frac{(\beta_k)^\theta}{c_k}=-\sum_{i=1,i\neq k}^{d}\left(\int_{\frac{\beta_i}{\beta_k}}^{+\infty}\lambda_U^{ik}\left(\frac{c_i}{c_k}t^{-\theta},1\right)dt-\eta\frac{\beta_k^{\theta-1}}{c_k}\beta_i\right), \forall k\in\{1,\ldots,d\}\/.
 	\end{equation}
 \end{Proposition}
 By solving the system (\ref{EqInt1}), we may obtain an equivalent of the extreme multivariate expectile, using the marginal quantiles.
 \begin{proof} 
 	The optimality system (\ref{S-functionsL2}) can be written in the following form
 	\begin{align*}
 	(2\alpha-1)\frac{\mathbb{E}[(X_k-x_k)_+]}{x_k\bar{F}_{X_k}(x_k)}-\frac{1-\alpha}{\bar{F}_{X_k}(x_k)}\left(1-\frac{\mathbb{E}[X_k]}{x_k}\right)&=\sum_{i=1,i\neq k}^{d}\left((1-\alpha)\frac{\mathbb{E}[(X_i-x_i)_-1\!\!1_{\{X_k<x_k\}}]}{x_k\bar{F}_{X_k}(x_k)}\right)\\	
 	&~~-\sum_{i=1,i\neq k}^{d}\alpha\frac{\mathbb{E}[(X_i-x_i)_+1\!\!1_{\{X_k>x_k\}}]}{x_k\bar{F}_{X_k}(x_k)},~~\forall k\in\{1,\ldots,d\}\/.
 	\end{align*}
 	For all $k\in\{1,\ldots,d\}$, we have (taking if necessary a subsequence)
 	$$\underset{\alpha\longrightarrow1}{\lim}(2\alpha-1)\frac{\mathbb{E}[(X_k-x_k)_+]}{x_k\bar{F}_{X_k}(x_k)}-\frac{1-\alpha}{\bar{F}_{X_k}(x_k)}\left(1-\frac{\mathbb{E}[X_k]}{x_k}\right)=\frac{1}{\theta-1}-\eta\frac{(\beta_k)^\theta}{c_k}\/,$$
 	and for all $i\in\{1,\ldots,d\}\setminus\{k\}$
 	\begin{align*}
 	\underset{\alpha\longrightarrow1}{\lim}(1-\alpha)\frac{\mathbb{E}[(X_i-x_i)_-1\!\!1_{\{X_k<x_k\}}]}{x_k\bar{F}_{X_k}(x_k)}&=\underset{\alpha\longrightarrow1}{\lim}\frac{1-\alpha}{\bar{F}_{X_k}(x_k)}\left(\frac{x_i}{x_k}\mathbb{P}(X_i<x_i,X_k<x_k)-\frac{\mathbb{E}[X_i1\!\!1_{\{X_i<x_i,X_k<x_k\}}]}{x_k}\right)\\ &=\eta\frac{\beta_k^{\theta}}{c_k}\frac{\beta_i}{\beta_k}=\eta\frac{\beta_k^{\theta-1}}{c_k}\beta_i\/.
 	\end{align*} 
 	Moreover,
 	\begin{align*}
 		\frac{\mathbb{E}[(X_i-x_i)_+1\!\!1_{\{X_k>x_k\}}]}{x_k\bar{F}_{X_k}(x_k)}&=\frac{1}{x_k\bar{F}_{X_k}(x_k)}\int_{x_i}^{+\infty}\mathbb{P}(X_i>t,X_k>x_k)dt\\
 		&=\intevariable{\frac{x_i}{x_k}}{+\infty}{\frac{\mathbb{P}(X_i>tx_k,X_k>x_k)}{\bar{F}_{X_k}(x_k)}dt}\\
 		&=\intevariable{\frac{x_i}{x_k}}{+\infty}{\frac{\mathbb{P}\left(\bar{F}_{X_i}(X_i)<\bar{F}_{X_i}(tx_k),\bar{F}_{X_k}(X_k)<\bar{F}_{X_k}(x_k)\right)}{\bar{F}_{X_k}(x_k)}dt}\/.
 		\end{align*}
 	Firstly, we remark that
 	\begin{align*}
 	\left|\intevariable{\frac{\beta_i}{\beta_k}}{\frac{x_i}{x_k}}{\frac{\mathbb{P}\left(\bar{F}_{X_i}(X_i)<\bar{F}_{X_i}(tx_k),\bar{F}_{X_k}(X_k)<\bar{F}_{X_k}(x_k)\right)}{\bar{F}_{X_k}(x_k)}dt}\right|&
 	\leq\left|\frac{x_i}{x_k}-\frac{\beta_i}{\beta_k}\right|\/.
 	\end{align*}
 	Since the functions $\lambda_U^{ik}$ are assumed to be continuous,
 	\begin{equation}\label{eq2-MRV}
 	\underset{\alpha\longrightarrow1}{\lim}\frac{\mathbb{P}\left(\bar{F}_{X_i}(X_i)<\bar{F}_{X_i}(tx_k),\bar{F}_{X_k}(X_k)<\bar{F}_{X_k}(x_k)\right)}{\bar{F}_{X_k}(x_k)}=\lambda_U^{ik}\left(\frac{c_i}{c_k}t^{-\theta},1\right)\/.
 	\end{equation}
 	In order to show that $$\underset{\alpha\longrightarrow 1}{\lim}\alpha\frac{\mathbb{E}[(X_i-x_i)_+1\!\!1_{\{X_k>x_k\}}]}{x_k\bar{F}_{X_k}(x_k)}=\int_{\frac{\beta_i}{\beta_k}}^{+\infty}\lambda_U^{ik}\left(\frac{c_i}{c_k}t^{-\theta},1\right)dt\/,$$
 	we may use the Lebesgue's Dominated Convergence Theorem with Potter's bounds (1942) (Lemma \ref{Potter}) for regularly varying functions.\\
 	First of all, 
 	\[ \frac{\mathbb{P}\left(\bar{F}_{X_i}(X_i)<\bar{F}_{X_i}(tx_k),\bar{F}_{X_k}(X_k)<\bar{F}_{X_k}(x_k)\right)}{\bar{F}_{X_k}(x_k)}\leq \min\left\{1,\frac{\bar{F}_{X_i}(tx_k)}{\bar{F}_{X_k}(x_k)}\right\} \/,\] 
 	since $\frac{\bar{F}_{X_i}(tx_k)}{\bar{F}_{X_k}(x_k)}=\frac{\bar{F}_{X_i}(tx_k)}{\bar{F}_{X_k}(tx_k)}\frac{\bar{F}_{X_k}(tx_k)}{\bar{F}_{X_k}(x_k)}$ and $\underset{\alpha\longrightarrow1}{\lim}\frac{\bar{F}_{X_i}(tx_k)}{\bar{F}_{X_k}(tx_k)}=\frac{c_i}{c_k}$, using Potter's bounds, for all $\varepsilon_1>0$ and $0<\varepsilon_2<\theta-1$, there exists $x_k^0(\varepsilon_2,\varepsilon_1)$  such that for  $\min\{x_k,tx_k\}\geq x_k^0(\varepsilon_2,\varepsilon_1)$
 	
	$$
 	\frac{\bar{F}_{X_i}(tx_k)}{\bar{F}_{X_k}(x_k)}\leq\left(\frac{c_i}{c_k}+2\varepsilon_1\right)t^{-\theta}\max(t^{\varepsilon_2},t^{-\varepsilon_2})\/.
 	$$
 	Lebesgue's theorem gives
	$$\underset{\alpha\longrightarrow1}{\lim}\intevariable{\frac{x_i}{x_k}}{+\infty}{\frac{\mathbb{P}\left(\bar{F}_{X_i}(X_i)<\bar{F}_{X_i}(tx_k),\bar{F}_{X_k}(X_k)<\bar{F}_{X_k}(x_k)\right)}{\bar{F}_{X_k}(x_k)}dt}=\int_{\frac{\beta_i}{\beta_k}}^{+\infty}{\lambda_U^{ik}\left(\frac{c_i}{c_k}t^{-\theta},1\right)dt}\/,$$
 	so, for all  $(i\neq k)\in\{1,\ldots,d\}^2$
 	$$\underset{\alpha\longrightarrow 1}{\lim}\frac{\mathbb{E}[(X_i-x_i)_+1\!\!1_{\{X_k>x_k\}}]}{x_k\bar{F}_{X_k}(x_k)}=\int_{\frac{\beta_i}{\beta_k}}^{+\infty}\lambda_U^{ik}\left(\frac{c_i}{c_k}t^{-\theta},1\right)dt\/.$$
 	Hence the system announced in this proposition.	
 \end{proof}
 In the general case of $\Sigma$-expectiles, with $\Sigma=(\pi_{ij})_{i\/,j=1\/,\ldots\/,d}$, $\pi_{ij}\geq  0$, $\pi_{ii}=\pi_i>0$, System (\ref{EqInt1}) becomes 
 \begin{equation*}
 \frac{1}{\theta-1}-\eta\frac{(\beta_k)^\theta}{c_k}=-\sum_{i=1,i\neq k}^{d}\frac{\pi_{ik}}{\pi_{k}}\left(\int_{\frac{\beta_i}{\beta_k}}^{+\infty}\lambda_U^{ik}\left(\frac{c_i}{c_k}t^{-\theta},1\right)dt-\eta\frac{\beta_k^{\theta-1}}{c_k}\beta_i\right), \forall k\in\{1,\ldots,d\}\/.
 \end{equation*}
 Moreover, let us remark that  System (\ref{EqInt1}) is  equivalent to the following system
 \begin{equation}\label{EqInt1-lim}
 \sum_{i=1}^{d}\int_{\frac{\beta_i}{\beta_k}}^{+\infty}\lambda_U^{ik}\left(c_it^{-\theta},c_k\beta_k^{-\theta}\right)dt=\sum_{i=1}^{d}\int_{\beta_i}^{+\infty}\lambda_U^{i1}\left(c_it^{-\theta},1\right)dt, \forall k\in\{2,\ldots,d\}\/.
 \end{equation}
 The limit points $\beta_i$ are thus completely determined by  the  asymptotic bivariate dependencies between the marginal components of the vector $\mathbf{X}$.
 \begin{Proposition}\label{PropInt1inf}
 	Assume that H1 and H2 are satisfied and the multivariate distribution of $\mathbf{X}$ is regularly varying in the sense of Definition \ref{MRVd1}, consider the $L_1$-expectiles $\mathbf{e}_\alpha(\mathbf{X})= (\mathbf{e}_\alpha^i(\mathbf{X}))_{i=1\/,\ldots\/,d}$. Then any limit vector $(\eta,\beta_2,\ldots,\beta_d)$ of $\left(\frac{1-\alpha}{\bar{F}_{X_1}(\mathbf{e}_\alpha^1(\mathbf{X}))},\frac{\mathbf{e}_\alpha^2(\mathbf{X})}{\mathbf{e}_\alpha^1(\mathbf{X})},\ldots,\frac{\mathbf{e}_\alpha^d(\mathbf{X})}{\mathbf{e}_\alpha^1(\mathbf{X})}\right)$ satisfies the following system of equations, $ \forall k\in\{1,\ldots,d\}$
 	\begin{equation}\label{EqInt2}
 	\frac{1}{\theta-1}-\eta\frac{(\beta_k)^\theta}{c_k}=-\sum_{i=1,i\neq k}^{d}\left(\frac{c_i}{c_k}\left(\frac{\beta_i}{\beta_k}\right)^{-\theta+1}\intevariable{1}{+\infty}{\lambda_U^{ik}\left(t^{-\theta},\frac{c_k}{c_i}\left(\frac{\beta_k}{\beta_i}\right)^{-\theta}\right)dt}-\eta\frac{\beta_k^{\theta-1}}{c_k}\beta_i\right)\/. 
 	\end{equation}
 \end{Proposition}	
 \begin{proof}
 	The proof is straightforward using a substitution in System (\ref{EqInt1}) and the positive homogeneity property of the bivariate tail dependence functions $\lambda_U^{ik}$ (see Proposition 2.2 in \cite{tailVine}).
 \end{proof}
 The main utility of writing the asymptotic optimality system in the form (\ref{EqInt2}) is the possibility to give an explicit form to $(\eta,\beta_2,\ldots,\beta_d)$  for some  dependence structures.
 \begin{description}
 	\item[Example] Consider that the dependence structure of $\mathbf{X}$ is given by an  Archimedean copula with generator $\psi$. The survival copula is given by  
 		$$
 		\bar{C}(x_1,\ldots,x_d)=\psi(\psi^{\leftharpoonup}(x_1)+\cdots+\psi^{\leftharpoonup}(x_d)) \/,
 		$$
 		where $\psi^{\leftharpoonup}(x)=\inf\{t\geq0| \psi(t)\leq x\}$ (see e.g. \cite{MultivArchi} for more details).  Assume that, $\psi$ is a regularly varying function with non-positive index  $\psi\in\mbox{RV}_{-\theta_\psi}$. According to \cite{charpentier2009}, the right tail dependence functions exist, and one can get their explicit forms $$\lambda_U^k(x_1,\ldots,x_k)=\left(\sum_{i=1}^{k}x_i^{-\frac{1}{\theta_\psi}}\right)^{-\theta_\psi} \/.$$
 	Thus, the bivariate upper tail dependence functions are given by 
 	$$\lambda_U^{ik}\left(t^{-\theta},\frac{c_k}{c_i}\left(\frac{\beta_k}{\beta_i}\right)^{-\theta}\right)=\left(t^{\frac{\theta}{\theta_\psi}}+\left(\frac{c_i}{c_k}\right)^{\frac{1}{\theta_\psi}}\left(\frac{\beta_k}{\beta_i}\right)^{\frac{\theta}{\theta_\psi}}\right)^{-\theta_\psi}\/.$$
 	In particular, if $\theta=\theta_\psi$, we have
 	$$\intevariable{1}{+\infty}{\lambda_U^{ik}\left(t^{-\theta},\frac{c_k}{c_i}\left(\frac{\beta_k}{\beta_i}\right)^{-\theta}\right)dt}=\frac{1}{\theta-1}\left(1+\left(\frac{c_i}{c_k}\right)^{\frac{1}{\theta}}\frac{\beta_k}{\beta_i}\right)^{-\theta+1}\/,$$
 	and System \ref{EqInt2} becomes 
 	\begin{equation*}
 	\frac{1}{\theta-1}-\eta\frac{(\beta_k)^\theta}{c_k}=-\sum_{i=1,i\neq k}^{d}\left(\frac{1}{\theta-1}\frac{c_i}{c_k}\left(\frac{\beta_i}{\beta_k}+\left(\frac{c_i}{c_k}\right)^{\frac{1}{\theta}}\right)^{-\theta+1}-\eta\frac{\beta_k^{\theta-1}}{c_k}\beta_i\right)\/. 
 	\end{equation*} \hfill$\Box$
  \end{description}
 \begin{lemma}[The comonotonic Fréchet case]\label{Frechet-Como}
 	Under H1 and H2, consider the $L_1$-expectiles $\mathbf{e}_\alpha(\mathbf{X})= (\mathbf{e}_\alpha^i(\mathbf{X}))_{i=1\/,\ldots\/,d}$. If $\mathbf{X}=(X_1,\ldots,X_d)$ is a comonotonic random vector, then the limit  $$(\eta,\beta_2,\ldots,\beta_d)=\underset{\alpha\longrightarrow1}{\lim}\left(\frac{1-\alpha}{\bar{F}_{X_1}(\mathbf{e}_\alpha^1(\mathbf{X}))},\frac{\mathbf{e}_\alpha^2(\mathbf{X})}{\mathbf{e}_\alpha^1(\mathbf{X})},\ldots,\frac{\mathbf{e}_\alpha^d(\mathbf{X})}{\mathbf{e}_\alpha^1(\mathbf{X})}\right)\/,$$
 	satisfies
 	$$\underset{\alpha\longrightarrow 1}{\lim}\frac{1-\alpha}{\bar{F}_{X_k}(\mathbf{e}_\alpha^k(\mathbf{X}))}=\frac{1}{\theta-1}~\mbox{ and }~\beta_k=c_k^{1/\theta},~~\forall k\in\{1,\ldots,d\}\/.$$ 
 \end{lemma}
 \begin{proof}
 	Since the random vector $\mathbf{X}$ is comonotonic, its survival copula is $$\overline{C}_{\mathbf{X}}(u_1,\ldots,u_d)=\min(u_1,\ldots,u_d), ~~\forall (u_1,\ldots,u_d)\in[0,1]^d\/.$$
 	We deduce the expression of the functions $\lambda_U^{ij}$ 
 	$$\lambda_U^{ij}(x_i,x_j)=\min(x_i,x_j), ~\forall (x_i,x_j)\in\mathbb{R}_+^2, \forall i,j\in\{1,\ldots,d\}\/.$$
 	So, 
 	\begin{align*}
 	\intevariable{1}{+\infty}{\lambda_U^{ik}\left(t^{-\theta},\frac{c_k}{c_i}(\frac{\beta_k}{\beta_i})^{-\theta}\right)dt}&=\intevariable{1}{+\infty}{\min\left(t^{-\theta},\frac{c_k}{c_i}\left(\frac{\beta_k}{\beta_i}\right)^{-\theta}\right)dt}\\
 	&=\left(\frac{\beta_k}{\beta_i}\left(\frac{c_k}{c_i}\right)^{-\frac{1}{\theta}}-1\right)_+\frac{c_k}{c_i}\left(\frac{\beta_k}{\beta_i}\right)^{-\theta}\\&~~+\frac{1}{\theta-1}\left(1+\left(\frac{\beta_k}{\beta_i}\left(\frac{c_k}{c_i}\right)^{-\frac{1}{\theta}}-1\right)_+\right)^{-\theta+1}\/.
 	\end{align*} 
 	Under assumptions H1 and H2, and by Proposition \ref{EqInt2}, let $(\eta,\beta_2,\ldots,\beta_d)$ be a solution of the following equation system.
 	\begin{align*}
 	\eta\sum_{i=1}^{d}\beta_i-\frac{1}{\theta-1}\sum_{i=1}^{d}c_i\beta^{-\theta+1}_i&=\sum_{i=1,i\neq k}^{d}c_k\beta_k^{-\theta}\beta_i\left(\frac{\beta_k}{\beta_i}\left(\frac{c_k}{c_i}\right)^{-\frac{1}{\theta}}-1\right)_+\\
 	&~+\frac{1}{\theta-1}\sum_{i=1,i\neq k}^{d}c_i\beta^{-\theta+1}_i\left[\left(1+\left(\frac{\beta_k}{\beta_i}\left(\frac{c_k}{c_i}\right)^{-\frac{1}{\theta}}-1\right)_+\right)^{-\theta+1}-1\right]\/, 
 	\end{align*}
 	$\forall k\in\{1,\ldots,d\}$. $\eta=\frac1{\theta-1}$ and $\beta_k = c_k^{\frac1\theta}$ is the only solution to this system.  	
 \end{proof}
 \begin{Proposition}[Asymptotic independence case]\label{Frechet-IndAsy}
 Under H1 and H2, consider the $L_1$-expectiles $\mathbf{e}_\alpha(\mathbf{X})= (\mathbf{e}_\alpha^i(\mathbf{X}))_{i=1\/,\ldots\/,d}$. If  $\mathbf{X}=(X_1,\ldots,X_d)$	is such that the pairs $(X_i\/,X_j)$ are asymptotically independent, then the limit vector $(\eta,\beta_2,\ldots,\beta_d)$ of  $\left(\frac{1-\alpha}{\bar{F}_{X_1}(\mathbf{e}_\alpha^1(\mathbf{X}))},\frac{\mathbf{e}_\alpha^2(\mathbf{X})}{\mathbf{e}_\alpha^1(\mathbf{X})},\ldots,\frac{\mathbf{e}_\alpha^d(\mathbf{X})}{\mathbf{e}_\alpha^1(\mathbf{X})}\right)$ satisfies
 	$$\eta=\frac{1}{(\theta-1)\left(1+\displaystyle\sum_{j=2}^{d}c_j^\frac{1}{\theta-1}\right)} 
 	\mbox{ and }
 	\beta_k=c_k^{\frac{1}{\theta-1}}
 	\/,$$
 	for all $k\in\{1,\ldots,d\}$. 
 \end{Proposition}
 \begin{proof}
The hypothesis of asymptotic bivariate independence means:
 	$$\underset{\alpha\longrightarrow1}{\lim}\frac{\mathbb{P}(X_i>x_i,X_j>x_j)}{\mathbb{P}(X_j>x_j)}=\underset{\alpha\longrightarrow1}{\lim}\frac{\mathbb{P}(X_i>tx_j,X_j>x_j)}{\mathbb{P}(X_j>x_j)}=0\/,$$
 	for all $(i,j)\in\{1,\ldots,d\}^2$ and for all $t>0$, then, Lebesgue's Theorem  used as in Proposition \ref{PropExpEx} gives
 	\begin{align*}
 	\underset{\alpha\longrightarrow 1}{\lim}\frac{\mathbb{E}[(X_i-x_i)_+1\!\!1_{\{X_j>x_j\}}]}{x_j\bar{F}_{X_j(x_j)}}&=\underset{\alpha\longrightarrow1}{\lim}\intevariable{\frac{x_i}{x_j}}{+\infty}{\frac{\mathbb{P}(X_i>tx_j,X_j>x_j)}{\mathbb{P}(X_j>x_j)}dt}\\
 	&=0\/.
 	\end{align*}
 	The extreme multivariate expectile  verifies  the following equation system 
 	\[ \frac{1}{\theta-1}-\frac{\eta}{c_k}\beta_k^{\theta}=+\sum_{i=1,i\neq k}^{d}\frac{\eta}{c_k}\beta_k^{\theta-1}\beta_i,~~ \forall k\in\{1,\ldots,d\}\/,\]
 	which can be rewritten as 
 	\begin{equation}\label{Eq-IndepAsy}
 	\frac{c_k}{\eta(\theta-1)\beta_k^{\theta-1}}=\sum_{i=1}^{d}\beta_i,~~\forall k\in\{1,\ldots,d\}\/,
 	\end{equation} 
 	hence $\beta_k=c_k^{\frac{1}{\theta-1}}$ for all $k\in\{1,\ldots,d\}$, and 
 	$$\eta=\frac{1}{(\theta-1)\left(1+\displaystyle\sum_{j=2}^{d}c_j^\frac{1}{\theta-1}\right)}\/.$$
 	
 \end{proof}
 In the general case of a matrix of positive coefficients 
 $\pi_{ij},~~ i,j\in\{1,\ldots,d\}$, the limits $\beta_i,i=2,\ldots,d$ remain the same, but the limit $\eta$ will change:
 $$\underset{\alpha\longrightarrow1}{\lim}\frac{\mathbf{e}_\alpha^k(\mathbf{X})}{\mathbf{e}_\alpha^1(\mathbf{X})}=c_k^{\frac{1}{\theta-1}}
 \mbox{ and }
 \underset{\alpha\longrightarrow 1}{\lim}\frac{1-\alpha}{\bar{F}_{X_k}(\mathbf{e}_\alpha^k(\mathbf{X}))}=\frac{c_k^{\frac{1}{\theta-1}}}{(\theta-1)\left(1+\displaystyle\sum_{j=2}^{d}\frac{\pi_{jk}}{\pi_{k}}c_j^\frac{1}{\theta-1}\right)}\/,$$
 for all $k\in\{1,\ldots,d\}$. \\
 We remark that
 $$\underset{\alpha\longrightarrow1}{\lim}\frac{1-\alpha}{\bar{F}_{X_i}(x_i)}\leq\frac{c_i^{\frac{1}{\theta-1}}}{\theta-1}\/,$$
 which allows a comparison between the marginal quantile and the corresponding component of the multivariate expectile, and since $F_{X_k}^{-1}(1-\cdot)$ is a regularly varying function at $0$ for all $k\in\{1,\ldots,d\}$ with index $-\frac{1}{\theta}$ (see Lemma~\ref{inevrseRV}), we get 
 $$e_\alpha^k(\mathbf{X})\underset{\alpha\longrightarrow1}{\sim}\mbox{VaR}_\alpha(X_k)\left(\theta-1\right)^{-\frac{1}{\theta}}\left(\frac{1+\displaystyle\sum_{i=2}^{d}c_i^\frac{1}{\theta-1}}{c_k^{\frac{1}{\theta-1}}}\right)^{-\frac{1}{\theta}}\/,$$
where $\mbox{VaR}_\alpha(X_k)$ denotes the Value at Risk of $X_k$ at level $\alpha$, ie the $\alpha$-quantile $F^{\leftarrow}_{X_k}(\alpha)$ of $X_k$. These conclusions coincide with the results obtained in dimension 1, for distributions that belong to the domain of attraction of Fréchet, in  \cite{belliniElena}. The values of constants $c_i$ determine the position of the marginal quantile compared to the corresponding component of the multivariate expectile for each risk.
 

 \section{Fréchet model with a dominant tail}
 This section is devoted to the case where $X_1$ has a dominant tail with respect to the $X_i$'s.
\begin{Proposition}[Asymptotic dominance]\label{Dom-Frechet}
	Under H1, consider the $L_1$-expectiles $\mathbf{e}_\alpha(\mathbf{X}) =(\mathbf{e}_\alpha^i(\mathbf{X}))_{i=1\/,\ldots\/,d}$. If  $$\underset{x\uparrow+\infty}{\lim}\frac{\bar{F}_{X_i}(x)}{\bar{F}_{X_1}(x)}=0,~~\forall i\in\{2,\ldots,d\}\/, \mbox{(dominant tail hypothesis)}$$
	then 
	$$\beta_i=\underset{\alpha\uparrow1}{\lim}\frac{\mathbf{e}_\alpha^i(\mathbf{X})}{\mathbf{e}_\alpha^1(\mathbf{X})}=0,~~~\underset{\alpha\uparrow1}{\lim}\frac{1-\alpha}{\bar{F}_{X_i}(\mathbf{e}_\alpha^i(\mathbf{X}))}=0,~~\forall i\in\{2,\ldots,d\}\/,$$
	and $$
	\underset{\alpha\uparrow1}{\lim}\frac{1-\alpha}{\bar{F}_{X_1}(\mathbf{e}_\alpha^1(\mathbf{X}))}=\frac{1}{\theta-1}\/.
	$$	
\end{Proposition}
The proof of Proposition \ref{Dom-Frechet} follows from  the following lemmas.
\begin{lemma}\label{L1-Dom-Frechet}
	Under H1, consider the $L_1$-expectiles $\mathbf{e}_\alpha(\mathbf{X}) =(\mathbf{e}_\alpha^i(\mathbf{X}))_{i=1\/,\ldots\/,d}$. If $$\underset{x\uparrow+\infty}{\lim}\frac{\bar{F}_{X_i}(x)}{\bar{F}_{X_1}(x)}=0,~~\forall i\in\{2,\ldots,d\}\/,$$
	then 
	$$\underset{\alpha\uparrow1}{\lim}\frac{\mathbb{E}[(X_i-x_i)_+1\!\!1_{\{X_1>x_1\}}]}{x_1\bar{F}_{X_1}(x_1)}=0,~\forall~i\in \{2,\ldots,d\}\/. $$	
\end{lemma}	 
\begin{lemma}\label{L2-Dom-Frechet}
	Under H1, consider the $L_1$-expectiles $\mathbf{e}_\alpha(\mathbf{X}) =(\mathbf{e}_\alpha^i(\mathbf{X}))_{i=1\/,\ldots\/,d}$. If $$\underset{x\uparrow+\infty}{\lim}\frac{\bar{F}_{X_i}(x)}{\bar{F}_{X_1}(x)}=0,~~\forall i\in\{2,\ldots,d\}\/,$$
	then
	$$0<\underset{\alpha\uparrow1}{\varliminf}\frac{1-\alpha}{\bar{F}_{X_1}(x_1)}\leq\underset{\alpha\uparrow1}{\varlimsup}\frac{1-\alpha}{\bar{F}_{X_1}(x_1)}<+\infty\/.$$
\end{lemma}	
The poofs of Lemmas \ref{L1-Dom-Frechet} and \ref{L2-Dom-Frechet} are given in Appendix \ref{pro4} and \ref{pro5} respectively.\\
Now, we have all necessary tools to prove Proposition \ref{Dom-Frechet}. 
\begin{proof}[Proof of Proposition \ref{Dom-Frechet}] 
	From Lemma \ref{L2-Dom-Frechet}, we have
	Taking if necessary a subsequence $(\alpha_n)_{n\in\mathbb{N}}$ $\alpha_n\longrightarrow 1$, we suppose that $\frac{1-\alpha}{\bar{F}_{X_1}(x_1)}$ is converging to a limit denoted $0<\eta<+\infty$ and that the limits  $\underset{\alpha\rightarrow1}{\lim}\frac{x_i}{x_1}=\beta_i$ exist.\\
	Going through the limit $(\alpha\rightarrow 1)$ in the $1^{\mbox{st}}$ equation of System \ref{eq1} divided by $x_1\bar{F}_{X_1}(x_1)$, leads using Lemma  \ref{L1-Dom-Frechet} to 
	\begin{equation}\label{DA-eq1}
	\underset{\alpha\uparrow1}{\lim}\left(\frac{1-\alpha}{\bar{F}_{X_1}(x_1)}\sum_{i\in J_{C}\cup J_{\infty}}^{}\frac{x_i}{x_1}\right)=\underset{\alpha\uparrow1}{\lim}\left(\eta\sum_{i\in J_{C}\cup J_{\infty}}^{}\frac{x_i}{x_1}\right)=\frac{1}{\theta-1}-\eta\/,
	\end{equation}
	we deduce that $J_{\infty}=\emptyset$.\\
	We suppose that $J_{C}\neq\emptyset$, so there exists at least one $i\in\{2,\ldots,d\}$ such that $i \in J_{C}$, and for all $j\in\{1,\ldots,d\}\backslash\{i\}$, we have
	$$\underset{\alpha\uparrow1}{\lim}\frac{\mathbb{E}[(X_j-x_j)_+1\!\!1_{\{X_i>x_i\}}]}{x_1\bar{F}_{X_1}(x_1)}=0\/.$$
	indeed, if $j\in J_{C}\backslash\{i\}$
	$$\underset{\alpha\uparrow1}{\lim}\frac{\mathbb{E}[(X_j-x_j)_+1\!\!1_{\{X_i>x_i\}}]}{x_1\bar{F}_{X_1}(x_1)}=\underset{\alpha\uparrow1}{\lim}\intevariable{\beta_j}{+\infty}{\frac{\mathbb{P}(X_j>tx_1,X_i>x_i)}{\mathbb{P}(X_1>x_1)}dt}\/,$$
	because 
	$$\frac{\mathbb{P}(X_j>tx_j,X_i>x_i)}{\mathbb{P}(X_1>x_1)}=\mathbb{P}(X_j>tx_j|X_i>x_i)\frac{\mathbb{P}(X_i>x_i)}{\mathbb{P}(X_1>x_1)}\/,$$
	and $\underset{\alpha\uparrow1}{\lim}\frac{\mathbb{P}(X_i>x_i)}{\mathbb{P}(X_1>x_1)}=\underset{\alpha\uparrow1}{\lim}\frac{\mathbb{P}(X_i>x_i)}{\mathbb{P}(X_1>x_i)}\frac{\mathbb{P}(X_1>x_i)}{\mathbb{P}(X_1>x_1)}=\underset{\alpha\uparrow1}{\lim}\frac{\mathbb{P}(X_i>x_i)}{\mathbb{P}(X_1>x_i)}\left(\frac{x_i}{x_1}\right)^{-\theta}=0$. And since 
	$$\frac{\mathbb{P}(X_j>tx_1,X_i>x_i)}{\mathbb{P}(X_1>x_1)}\leq \min\left(\frac{ \mathbb{P}(X_j>tx_1)}{\mathbb{P}(X_1>x_1)},\frac{\mathbb{P}(X_i>x_i)}{\mathbb{P}(X_1>x_1)}\right)\/,$$
	and using Potter's Bounds associated to $\bar{F}_{X_1}$ as regularly varying function in order to apply the dominated convergence Theorem, we get
	$$\underset{\alpha\uparrow1}{\lim}\frac{\mathbb{E}[(X_j-x_j)_+1\!\!1_{\{X_i>x_i\}}]}{x_1\bar{F}_{X_1}(x_1)}=\intevariable{\beta_j}{+\infty}{\underset{\alpha\uparrow1}{\lim}\frac{\mathbb{P}(X_j>tx_1,X_i>x_i)}{\mathbb{P}(X_1>x_1)}dt}=0,~\forall j\in J_{C}\backslash\{i\}\/.$$
	Now, if $j\in J_{0}$ then
	\begin{align*}
	\frac{\mathbb{E}[(X_j-x_j)_+1\!\!1_{\{X_i>x_i\}}]}{x_1\bar{F}_{X_1}(x_1)}&=\intevariable{x_j}{x_1}{\frac{\mathbb{P}(X_j>t,X_i>x_i)}{x_1\mathbb{P}(X_1>x_1)}dt}+\intevariable{1}{+\infty}{\frac{\mathbb{P}(X_j>tx_1,X_i>x_i)}{\mathbb{P}(X_1>x_1)}dt}\\
	&\leq\left(1-\frac{x_j}{x_1}\right)\frac{\mathbb{P}(X_i>x_i)}{\mathbb{P}(X_1>x_1)}+\intevariable{1}{+\infty}{\frac{\mathbb{P}(X_j>tx_1,X_i>x_i)}{\mathbb{P}(X_1>x_1)}dt}\/,
	\end{align*}
	thus
	$$\underset{\alpha\uparrow1}{\lim}\frac{\mathbb{E}[(X_j-x_j)_+1\!\!1_{\{X_i>x_i\}}]}{x_1\bar{F}_{X_1}(x_1)}=0, ~\forall j\in J_{0}\/,$$
	because $\underset{\alpha\uparrow1}{\lim}\frac{\mathbb{P}(X_i>x_i)}{\mathbb{P}(X_1>x_1)}=0$ and $\underset{\alpha\uparrow1}{\lim}\intevariable{1}{+\infty}{\frac{\mathbb{P}(X_j>tx_1,X_i>x_i)}{\mathbb{P}(X_1>x_1)}dt}=0$.\\
	Going through the limit $(\alpha\rightarrow 1)$ in the  $i^{\mbox{th}}$ equation of System \ref{eq1} divided by $x_1\bar{F}_{X_1}(x_1)$, leads to
	$$-\eta\beta_i=\eta(1+\sum_{j\in J_{C}\backslash\{i\}}^{}\beta_j)\/,$$
	which is absurd, and consequently $J_{C}=\emptyset$.\\
	We have thus proved that $J_{0}=\{2,\ldots,d\}$ which means
	$$\underset{\alpha\uparrow1}{\lim}\frac{\mathbf{e}_\alpha^i(\mathbf{X})}{\mathbf{e}_\alpha^1(\mathbf{X})}=\beta_i=0,~~\forall i\in\{2,\ldots,d\}\/.$$
	And from Equation \ref{DA-eq1} we deduce also that
	$$\eta=\underset{\alpha\uparrow1}{\lim}\frac{1-\alpha}{\bar{F}_{X_1}(\mathbf{e}_\alpha^1(\mathbf{X}))}=\frac{1}{\theta-1}\/,$$
	and by Lemma \ref{Lemme+inf} that 
	$$\underset{\alpha\uparrow1}{\lim}\frac{1-\alpha}{\bar{F}_{X_k}(\mathbf{e}_\alpha^i(\mathbf{X}))}=0,~~\forall i\in\{2,\ldots,d\}\/.$$
\end{proof}
Proposition \ref{Dom-Frechet} shows that the dominant risk behaves asymptotically as in the univariate case, and its component in the extreme multivariate expectile satisfies
$$\mathbf{e}_\alpha^1(\mathbf{X})\overset{\alpha\uparrow1}{\sim}(\theta-1)^{-\frac{1}{\theta}}\mbox{VaR}_{\alpha}(X_1)\overset{\alpha\uparrow1}{\sim}\mathbf{e}_\alpha(X_1)\/,$$
the right equivalence is proved, in the univariate case, in  \cite{belliniElena}, Proposition 2.3. 
\begin{description}
	\item[Example] Consider Pareto distributions,
	$X_i\sim Pa(a_i,b),i=1,\ldots,d$, such that $a_i>a_1$ 
	for all $i\in\{1,\ldots,d\}$. The tail of $X_1$ dominates that of the $X_i$'s and Proposition \ref{Dom-Frechet} applies.
\end{description}


\section{Estimation of the extreme expectiles}
In this section, we propose some estimators of the extreme multivariate expectile. We focus on the cases of asymptotic independence and comonotonicity, for which the equation system is more tractable. We begin with the main ideas of our approach, then, we construct the estimators using the extreme values statistical tools and prove its consistency. We terminate this section with a simulation study.  
\begin{Proposition}[Estimation's idea]\label{Id-Est}
	Using notations of previous sections, consider the $L_1$-expectiles $\mathbf{e}_\alpha(\mathbf{X}) =(\mathbf{e}_\alpha^i(\mathbf{X}))_{i=1\/,\ldots\/,d}$. Under  H1, H2 and the assumption that the vector $\left(\frac{1-\alpha}{\bar{F}_{X_1}(\mathbf{e}_\alpha^1(\mathbf{X}))},\frac{\mathbf{e}_\alpha^2(\mathbf{X})}{\mathbf{e}_\alpha^1(\mathbf{X})},\ldots,\frac{\mathbf{e}_\alpha^d(\mathbf{X})}{\mathbf{e}_\alpha^1(\mathbf{X})}\right)$  has a unique limit point $(\eta\/,\beta_2\/,\ldots\/,\beta_d)$,
	$$\mathbf{e}_\alpha(\mathbf{X})\underset{\alpha\longrightarrow1}{\sim}\mbox{VaR}_\alpha(X_1)\eta^{\frac{1}{\theta}}(1,\beta_2,\ldots,\beta_d)^T\/.$$ 
\end{Proposition} 
\begin{proof}[Proof]
	Let $(\eta,\beta_2,\ldots,\beta_d)=\lim_{\alpha\rightarrow 1} \left(\frac{1-\alpha}{\bar{F}_{X_1}(\mathbf{e}_\alpha^1(\mathbf{X}))},\frac{\mathbf{e}_\alpha^2(\mathbf{X})}{\mathbf{e}_\alpha^1(\mathbf{X})},\ldots,\frac{\mathbf{e}_\alpha^d(\mathbf{X})}{\mathbf{e}_\alpha^1(\mathbf{X})}\right)$, we have
	$$\mathbf{e}_\alpha(\mathbf{X})\underset{\alpha\longrightarrow1}{\sim} \mathbf{e}_\alpha^1(\mathbf{X})(1,\beta_2,\ldots,\beta_d)^T\/.$$
	Moreover,
	$\underset{\alpha\longrightarrow1}{\lim}\frac{1-\alpha}{\bar{F}_{X_1}(\mathbf{e}_\alpha^1(\mathbf{X}))}=\eta$, and Theorem 1.5.12 in \cite{bingham1989regular} states that  $F_{X_1}^{\leftarrow}(1-.)$ is regularly varying at $0$, with index $-\frac1\theta$. This leads to
	$$\mathbf{e}_\alpha^1(\mathbf{X})\underset{\alpha\longrightarrow1}{\sim}F_{X_1}^{\leftarrow}(\alpha)\left(\frac{1}{\eta}\right)^{-\frac{1}{\theta}}\/,$$
	and the result follows. 
\end{proof}  
 Proposition \ref{Id-Est} gives a way to estimate the extreme multivariate expectile. Let  $\mathbf{X}=(X_1,\ldots,X_d)^T$ be  an independent sample of size $n$ of $\mathbf{X}$, with $\mathbf{X}_i=(X_{1,i},\ldots,X_{d,i})^T$ for all $i\in\{1,\ldots,n\}$. We denote by $X_{i,1,n}\leq X_{i,2,n}\leq\cdots\leq X_{i,n,n}$ the  ordered sample corresponding to $X_i$.
 \subsection{Estimator's construction}  
 We begin with the case of asymptotic independence.  Propositions \ref{Frechet-IndAsy} and \ref{Id-Est} are the key tools in the construction of the estimator.  
 We have for all $i\in\{1,\ldots,d\}$
 $$\beta_i=c_i^{\frac{1}{\theta-1}}\/,\ \mbox{and} \ \underset{\alpha\longrightarrow 1}{\lim}\frac{1-\alpha}{\bar{F}_{X_i}(\mathbf{e}_\alpha^i(\mathbf{X}))}=\frac{c_i^{\frac{1}{\theta-1}}}{(\theta-1)\left(1+\displaystyle\sum_{j=2}^{d}c_j^\frac{1}{\theta-1}\right)}\/.$$
 Proposition \ref{Id-Est} gives
 $$\mathbf{e}_\alpha(\mathbf{X})\overset{\alpha\uparrow1}{\sim}\mbox{VaR}_\alpha(X_1)\left(\theta-1\right)^{-\frac{1}{\theta}}\left(1+\sum_{i=2}^{d}c_i^\frac{1}{\theta-1}\right)^{-\frac{1}{\theta}}\left(1,c_2^{\frac{1}{\theta-1}},\ldots,c_d^{\frac{1}{\theta-1}}\right)^T\/.$$
 So, in order to estimate the extreme multivariate expectile, we need an estimator of the univariate quantile of $X_1$, of the tail equivalence parameters. and of $\theta$.\\ 
 In the same way, and for the case of comonotonic risks, we may use 
 Proposition \ref{Frechet-Como}
 $$\underset{\alpha\longrightarrow 1}{\lim}\frac{1-\alpha}{\bar{F}_{X_i}(\mathbf{e}_\alpha^i(\mathbf{X}))}=\frac{1}{\theta-1}~\mbox{ and }~\beta_i=c_i^{1/\theta},~~\forall i\in\{1,\ldots,d\}\/,$$
 and by Proposition \ref{Id-Est} we obtain
 $$\mathbf{e}_\alpha(\mathbf{X})^T\overset{\alpha\uparrow1}{\sim}\mbox{VaR}_\alpha(X_1)\left(\theta-1\right)^{-\frac{1}{\theta}}(1,c_2^{\frac{1}{\theta}},\ldots,c_d^{\frac{1}{\theta}})^T\/.$$ 
 The $X_i$'s have all the same index $\theta$ of regular variation, which is also the same as the index of regular variation of $\Vert \mathbf{X}\Vert$. We propose to estimate $\theta$  by using the Hill estimator $\widehat{\gamma}$. We shall denote $\widehat{\theta}=\frac1{\widehat{\gamma}}$.  See \cite{hill1975simple} for details on the Hill estimator. 
 In order to estimate the $c_i$'s, we shall use the GPD approximation: for $u$ a large threshold, and $x\geq u$,
 $$\bar{F}(x)\sim \bar{F}(u)\left(\frac{x}{u}\right)^{-\theta}\/.$$
 Let $k\in\mathbb{N}$ be fixed and consider the thresholds $u_i$:
 $$\bar{F}_{X_i}(u_i)=\bar{F}_{X_1}(u_1)=\frac{k}{n},~~\forall i\in\{1,\ldots,d\}\/.$$  
 The $u_i$ are estimated by $X_{i,n-k+1,n}$ with $k\rightarrow\infty$ and $k/n\rightarrow0$ as $n\rightarrow\infty$. Using Lemma \ref{Lemme+inf}, we get    $$c_i= \lim_{n\rightarrow \infty}\left(\frac{u_i}{u_1}\right)^\theta\/.$$
 	We shall consider
 \begin{equation}\label{Frechet-Est-ciPot}
 \hat{c}_i=\left(\frac{X_{i,n-k+1,n}}{X_{1,n-k+1,n}}\right)^\frac{1}{\hat{\gamma}(k)}\/,
 \end{equation} 
 where $\hat{\gamma}(k)$ is the Hill's estimator of the extreme values index constructed using the $k$ largest observations of $\Vert\mathbf{X}\Vert$. Let $\widehat{\theta}=\frac1{\widehat{\gamma}(k)}$.\\
 	\begin{Proposition}\label{prop:c}
 		Let $k=k(n)$ be such that $k\rightarrow \infty$ and $k/n\longrightarrow 0$ as $n\rightarrow \infty$. Under H1 and H2, for any $i=2\/,\ldots\/,d$,
 		$$\widehat{c}_i \stackrel{\mathbb{P}}{\longrightarrow} c_i \/.$$
 	\end{Proposition}
 	\begin{proof}
 		The results in \cite{segers} page 86 imply that for any $i=1\/,\ldots\/,d$
 		$$\frac{X_{i,n-k+1,n}}{u_i} \stackrel{\mathbb{P}}{\longrightarrow} 1\/.$$
 		Moreover, it is well known (see \cite{hill1975simple}) that the Hill estimator is consistent. Using (\ref{Frechet-Est-ciPot}), and the fact that 
 		$$\frac{X_{i,n-k+1,n}}{X_{1,n-k+1,n}} \sim \frac{u_i}{u_1} \ \mbox{in probability and thus is bounded in probability}\/,$$
 		we get the result.
 	\end{proof}
 
 To estimate the extreme quantile, we will use Weissman's estimator (1978) \cite{weissman1978}: 
 $$\widehat{\mbox{VaR}}_\alpha(X_1)=X_{1,n-k(n)+1,n}\left(\frac{k(n)}{(1-\alpha) n}\right)^{\hat{\gamma}}\/.$$
 The properties of Weissman's estimator are presented in Embrechts et al. (1997)~\cite{embrechts1997} and also in \cite{segers} page 119. In order to prove the consistency of our estimators of 
extreme multivariate expectiles, we shall need the following second order condition (see \cite{segers} Section 4.4). 
\begin{definition}
 A random variable $X$ satisfying H1 with $\theta = \frac1\gamma>0$ will be said to verify the second order condition $\rm{SOC}_{-\beta}(b)$ with $\beta>0$ and $b\in \mathcal{RV}_{-\beta}(+\infty)$ 
if the function $U~: y \leadsto F^{\leftarrow}(1-\frac1y)$ satisfies for $u>0$:
$$\frac{U(ux)}{U(x)} = u^\gamma \left(1+h_{-\beta}(u) b(x) +o(b(x))\right)\/, \ \mbox{as} \ x \ \mbox{goes to infinity} \/.$$
where $h_{-\beta}(u) = \frac{1-u^{-\beta}}{\beta}$.  
\end{definition}

  Now, we can deduce some estimators of the extreme multivariate expectile, using the previous ones, in the cases of asymptotic independence and perfect dependence. 
  \begin{definition}[Multivariate expectile estimator, Asymptotic independence]\label{Frechte-Estimateur}
  	Under H1 and H2, in the case of bivariate asymptotic independence of the random vector $\mathbf{X}=(X_1,\ldots,X_d)^T$, we define the estimator of the $L_1$-expectile as follows
  \begin{eqnarray*}
         \hat{\mathbf{e}}^{\bot}_{\alpha}(\mathbf{X})&=&X_{1,n-k(n)+1,n}\left(\frac{k(n)}{(1-\alpha) 
n}\right)^{\hat{\gamma}}\left(\frac{\hat{\gamma}}{1-\hat{\gamma}}\right)^{\hat{\gamma}}\left(\frac{1}{1+\sum_{\ell=2}^{d}\hat{c_\ell}^{\frac{\hat{\gamma}}{1-\hat{\gamma}}}}\right)^{\hat{\gamma}}
\left(1,\hat{c}_2^{\frac{\hat{\gamma}}{1-\hat{\gamma}}},\ldots,\hat{c}_d^{\frac{\hat{\gamma}}{1-\hat{\gamma}}}\right)^T\\
&=& \widehat{\mbox{VaR}}_\alpha(X_1) 
\left(\frac{\hat{\gamma}}{1-\hat{\gamma}}\right)^{\hat{\gamma}}\left(\frac{1}{1+\sum_{\ell=2}^{d}\hat{c_\ell}^{\frac{\hat{\gamma}}{1-\hat{\gamma}}}}\right)^{\hat{\gamma}}
\left(1,\hat{c}_2^{\frac{\hat{\gamma}}{1-\hat{\gamma}}},\ldots,\hat{c}_d^{\frac{\hat{\gamma}}{1-\hat{\gamma}}}\right)^T\/.
        \end{eqnarray*}
        
  \end{definition}
\begin{definition}[Multivariate expectile estimator, comonotonic risks]\label{Frechte-Estimateur-CC}
	Under the assumptions of the Fréchet model with equivalent tails, for a comonotonic random vector $\mathbf{X}=(X_1,\ldots,X_d)^T$, we define the estimator of $L_1$-expectile as follows
\begin{eqnarray*}
 \hat{\mathbf{e}}^{+}_{\alpha}(\mathbf{X})&=&X_{1,n-k(n)+1,n}\left(\frac{k(n)}{(1-\alpha) 
n}\right)^{\hat{\gamma}}\left(\frac{\hat{\gamma}}{1-\hat{\gamma}}\right)^{\hat{\gamma}}\left(1,\hat{c}_2^{\hat{\gamma}},\ldots,\hat{c}_d^{\hat{\gamma}}\right)^T\\
&=& \widehat{\mbox{VaR}}_\alpha(X_1) \left(\frac{\hat{\gamma}}{1-\hat{\gamma}}\right)^{\hat{\gamma}}\left(1,\hat{c}_2^{\hat{\gamma}},\ldots,\hat{c}_d^{\hat{\gamma}}\right)^T\/.
\end{eqnarray*}

\end{definition} 
We prove below that if the second order condition $\rm{SOC}_{-\beta}(b)$ is satisfied, then the term by term ratio  
$\hat{\mathbf{e}}^{\bot}_{\alpha}(\mathbf{X}) /{\mathbf{e}}_{\alpha}(\mathbf{X}) $ goes to $1$ in probability in the asymptotically independent case and $\hat{\mathbf{e}}^{+}_{\alpha}(\mathbf{X}) 
/{\mathbf{e}}_{\alpha}(\mathbf{X}) $ goes to $1$ in probability in the comonotonic case. More work is required to get the asymptotic normality.    
\begin{theorem}
Assume that H1, H2 and $\rm{SOC}_{-\beta}(b)$ are satisfied. Choose $k=k(n)$ such that 
\begin{itemize}
 \item $k(n) \rightarrow \infty$ as $n\rightarrow \infty$,
 \item $k(n)/n \rightarrow 0$ as $n\rightarrow \infty$,
 \item $\sqrt{k(n)} \left(1+ \log^2 \frac{k(n)}{n(1-\alpha)}\right)^{-\frac12} \rightarrow \infty$ as $n\rightarrow \infty$.
\end{itemize}
Then, if each pair of the random vector $\mathbf{X}$ is asymptotically independent,
$$\hat{\mathbf{e}}^{\bot}_{\alpha}(\mathbf{X}) /{\mathbf{e}}_{\alpha}(\mathbf{X})  \longrightarrow 1 \ \mbox{in probability, as} \ n\rightarrow \infty\/.$$
If the random vector $\mathbf{X}$ is comonotonic, then 
$$\hat{\mathbf{e}}^{+}_{\alpha}(\mathbf{X}) /{\mathbf{e}}_{\alpha}(\mathbf{X})\longrightarrow 1 \ \mbox{in probability, as} \ n\rightarrow \infty\/.$$
\end{theorem}
\begin{proof}
With the  $\rm{SOC}_{-\beta}(b)$ hypothesis and the choice of $k$, we get by using (4.18) p.120 in \cite{segers} that
$$\frac{\widehat{\mbox{VaR}}_\alpha(X_1)}{\mbox{VaR}_\alpha(X_1)} \longrightarrow 1 \ \mbox{in probability as} \ n \rightarrow \infty\/.$$
Then the announced results follow from Propositions \ref{Frechet-IndAsy} and \ref{prop:c}. 
\end{proof}
 \subsection{Numerical illustration}
 The attraction domain of Fréchet contains the usual distributions of Pareto, Student, Burr and Cauchy. In order to illustrate the convergence of the proposed estimators, we study numerically, the cases of Pareto, Burr and Student distributions.\\
 
 In the independence case which is a special case of asymptotic independence, the functions $l^\alpha_{X_i,X_j}$ defined in (\ref{DefFL}) have the following expression
 $$l^\alpha_{X_i,X_j}(x_i,x_j)=\alpha \left(\bar{F}_{X_j}(x_j)\mathbb{E}\left[\left(X_i-x_i\right)_+\right]\right)-(1-\alpha)\left(F_{X_j}(x_j)\mathbb{E}\left[\left(X_i-x_i\right)_-\right]\right)\/.$$
 In the comonotonic case we have 
 \begin{align*}
 l^\alpha_{X_i,X_j}(x_i,x_j)&=\alpha \left(\bar{F}_{X_j}(x_j)(\mu_{i,j}-x_i)_++\mathbb{E}\left[\left(X_i-\max(x_i,\mu_{i,j})\right)_+\right]\right)\\&~~-(1-\alpha)\left(F_{X_j}(x_j)(x_i-\mu_{i,j})_++\mathbb{E}\left[\left(X_i-\min(x_i,\mu_{i,j})\right)_-\right]\right)\/,
 \end{align*}
 where $\mu_{i,j}=F^{\longleftarrow}_{X_i}(F_{X_j}(x_j))$.\\
 From these expressions, the exact value of the extreme multivariate expectile is obtained using numerical optimization, and we can confront it to the estimated values. The choice of  $k(n)$ is function of the distributions parameters, and it is done in our simulations using graphical illustrations. We present the estimators for different values of $k(n)$ that belong to the common convergence range of the estimators of tail equivalence coefficients, in order to verify the stability of the expectile estimator?s convergence.
 \subsubsection{Pareto distributions}
 We consider a bivariate Pareto model $X_i\sim Pa(a,b_i),~i\in\{1,2\}$. Both distributions have the same scale parameter $a$, so they have equivalent tails with equivalence parameter  
 $$c_2=\underset{x\rightarrow+\infty}{\lim}\frac{\bar{F}_{X_2}(x)}{\bar{F}_{X_1}(x)}=\underset{x\rightarrow+\infty}{\lim}\frac{\left(\frac{b_2}{b_2+x}\right)^a}{\left(\frac{b_1}{b_1+x}\right)^a}=\left(\frac{b_2}{b_1}\right)^a\/.$$   
In what follows, we consider two models for which the exact values of the $L_1$-expectiles are computable. In the first model, the $X_i$'s are independent. In the second one, the $X_i$'s are comonotonic and for Pareto distributions $\mu_{i,j}=\frac{b_i}{b_j}x_j$. In the simulations below, we have taken the same $k=k(n)$ to get $\hat{\gamma}$ and $\hat{\mathbf{e}}_{\alpha}(\mathbf{X})$. For $X_i\sim Pa(2,5\cdot(i+1))_{i=1,2}$, and $n=100000$, 
 Figure \ref{Estim-Frechet1} illustrates the convergence of estimator $\hat{c}_2$. On the left, the shaded area indicates suitable values of $k(n)$ for $n=100000$. The boxplots are obtained for different values of $n$ and a fixed $k\in k(n)$, the data size is 1000.
 \begin{figure}[H]
 	\center
 	\includegraphics[width=8cm]{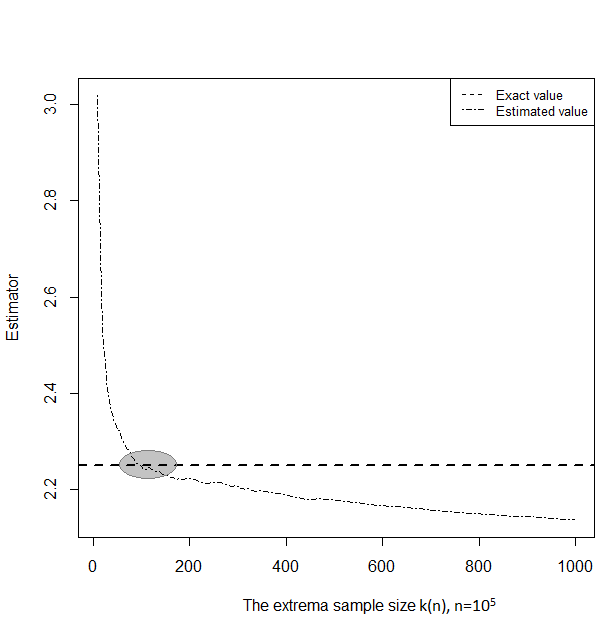} 
 	\includegraphics[width=8cm]{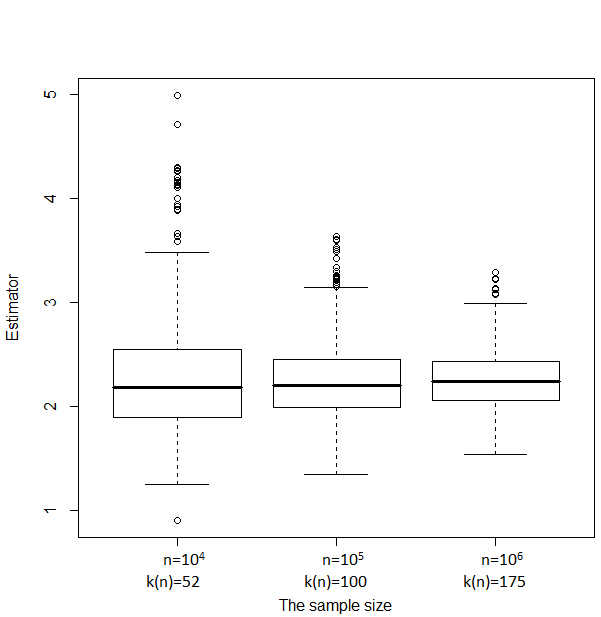}
 	\caption{Convergence of $\hat{c}_2$. $X_i\sim Pa(2,5\cdot(i+1))_{i=1,2}$.}
 	\label{Estim-Frechet1}
 \end{figure}  
Figure \ref{Conv-Exp-F} presents the obtained results for different $k(n)$ values in the independence case where $n=100000$. A multivariate illustration in dimension 4 is given in Figure \ref{Conv-Exp-Fd4}. The comonotonic case is illustrated in Figure \ref{Conv-Exp-F-C}. The simulations parameters are $a=2$, $b_1=10$ and $b_2=15$.
 \begin{figure}[H]
 	\center
 	\includegraphics[width=8cm]{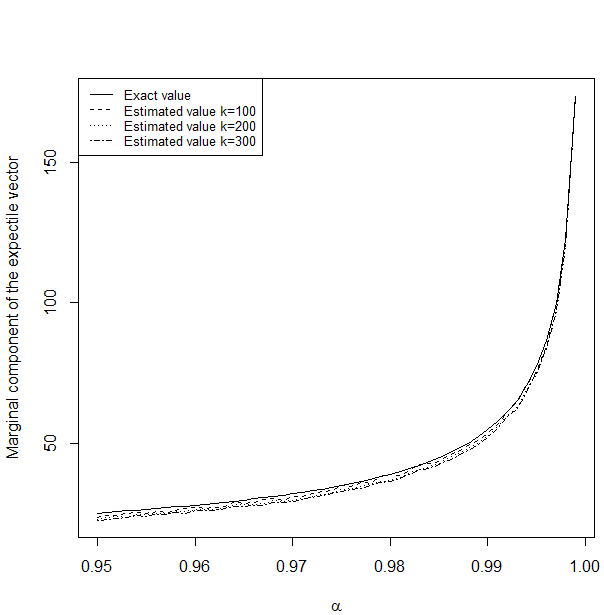} 
 	\includegraphics[width=8cm]{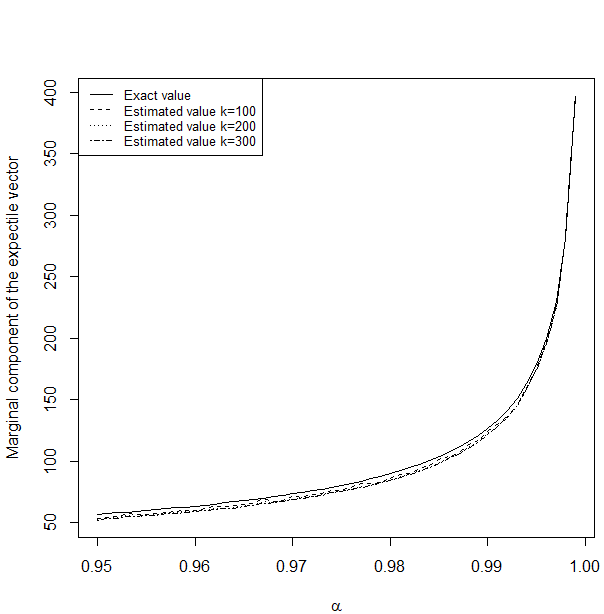}
 	\caption{Convergence of $\hat{\mathbf{e}}^{\bot}_{\alpha}(\mathbf{X})$ (asymptotic independence case). On the left, the first coordinate of ${\mathbf{e}}_{\alpha}(\mathbf{X})$ and $\hat{\mathbf{e}}^{\bot}_{\alpha}(\mathbf{X})$ for various values of $k=k(n)$ are plotted. The right figure concerns the second coordinate.}
 	\label{Conv-Exp-F}
 \end{figure}
\begin{figure}[H]
	\center
	\includegraphics[width=8cm]{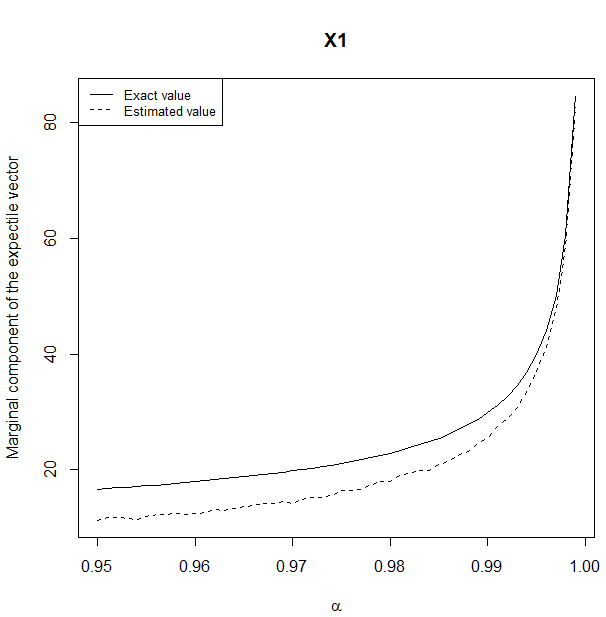} 
	\includegraphics[width=8cm]{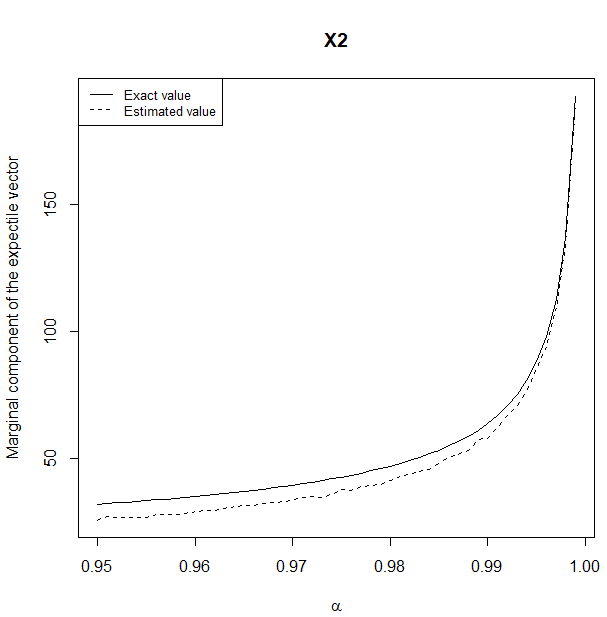} 
	\includegraphics[width=8cm]{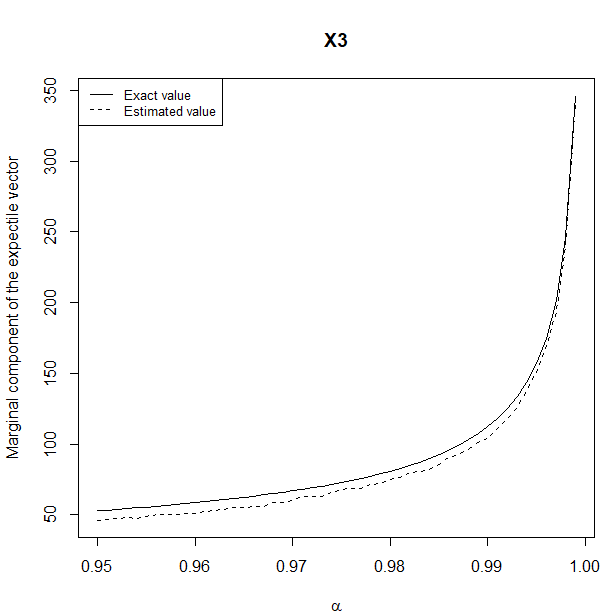} 
	\includegraphics[width=8cm]{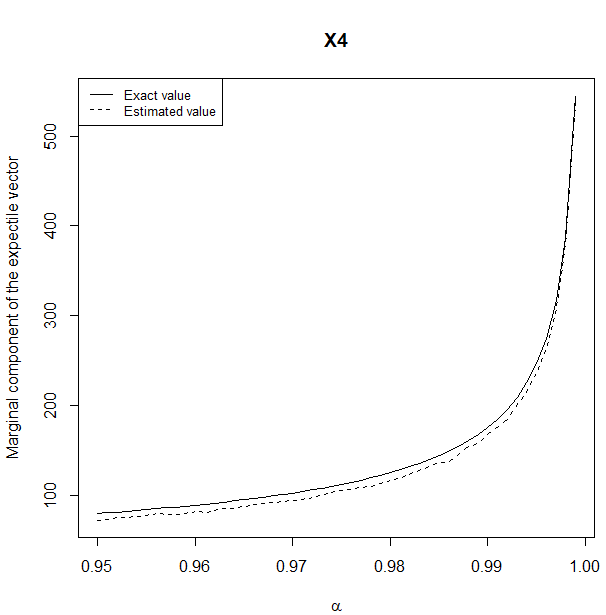} 
	\caption{Convergence of $\hat{\mathbf{e}}^{\bot}_{\alpha}(\mathbf{X})$ (asymptotic independence case). The coordinates of ${\mathbf{e}}_{\alpha}(\mathbf{X})$ and $\hat{\mathbf{e}}^{\bot}_{\alpha}(\mathbf{X})$ in dimension $d=4$, $n=100000$ and $k(n)=100$.}
	\label{Conv-Exp-Fd4}
\end{figure}
  \begin{figure}[H]
  	\center
  	\includegraphics[width=8cm]{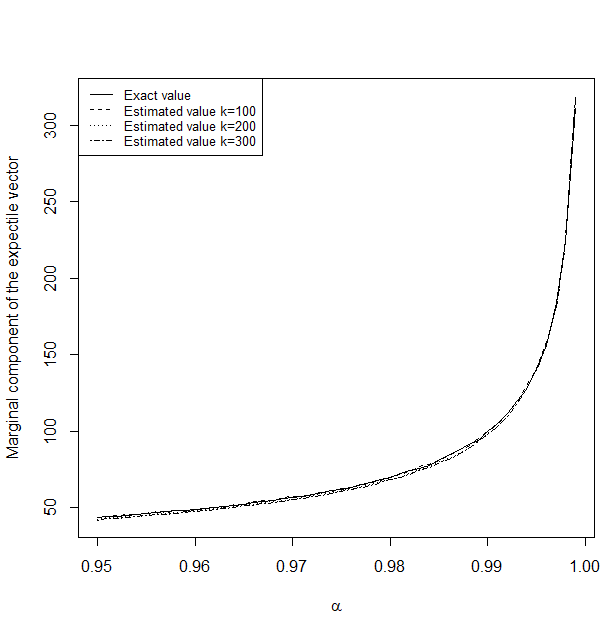} 
  	\includegraphics[width=8cm]{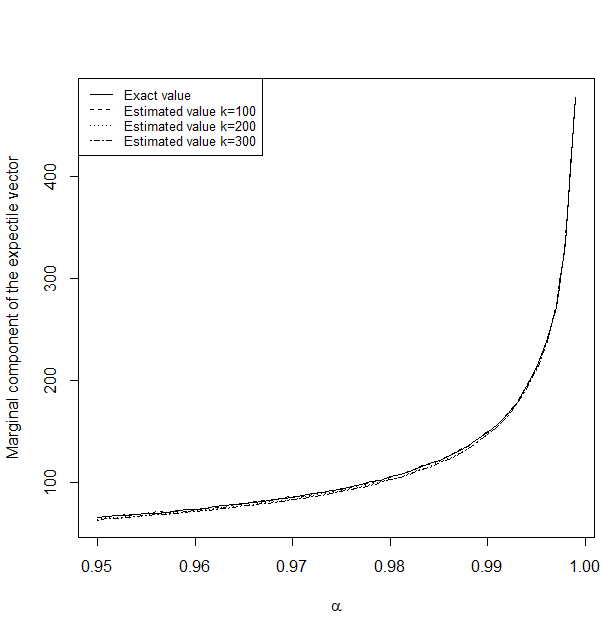}
  	\caption{Convergence of $\hat{\mathbf{e}}^{+}_{\alpha}(\mathbf{X})$ (comonotonic case). On the left, the first coordinate of ${\mathbf{e}}_{\alpha}(\mathbf{X})$ and $\hat{\mathbf{e}}^{\bot}_{\alpha}(\mathbf{X})$ for various values of $k=k(n)$ are plotted. The right figure concerns the second coordinate.}
  	\label{Conv-Exp-F-C}
  \end{figure}
\subsubsection{Burr distributions}
 We consider a multivariate Burr model $X_i\sim Burr(a,b_i,\tau),~i\in\{1,\dots,d\}$. In this case, the tails are equivalents with equivalence parameter  
$$c_i=\underset{x\rightarrow+\infty}{\lim}\frac{\bar{F}_{X_i}(x)}{\bar{F}_{X_1}(x)}=\underset{x\rightarrow+\infty}{\lim}\frac{\left(\frac{b_i}{b_i+x^{\tau}}\right)^a}{\left(\frac{b_1}{b_1+x^{\tau}}\right)^a}=\left(\frac{b_i}{b_1}\right)^a, \forall i\in\{2,\ldots,d\}\/,$$   
and $\bar{F}_{X_i}\in RV_{a*\tau}(+\infty)$ for all $i$ in $\{1,\ldots,d\}$. In the Burr comonotonic case $\mu_{i,j}=\left(\frac{b_i}{b_j}\right)^{\frac{1}{\tau}}x_j$. The model is asymptotically equivalent to the Pareto one, but the margins are different, which helps to test the pertinence of the estimation processes compared to the theoretical resultants. Figures \ref{BurrI} and \ref{BurrC} present the obtained results for different $k(n)$ values in the independence and the comonotonic cases respectively. The simulations parameters are $a=4$, $b_1=10$, $b_2=15$, $\tau=0.75$ and $n=10000$.
\begin{figure}[H]
	\center
	\includegraphics[width=8cm]{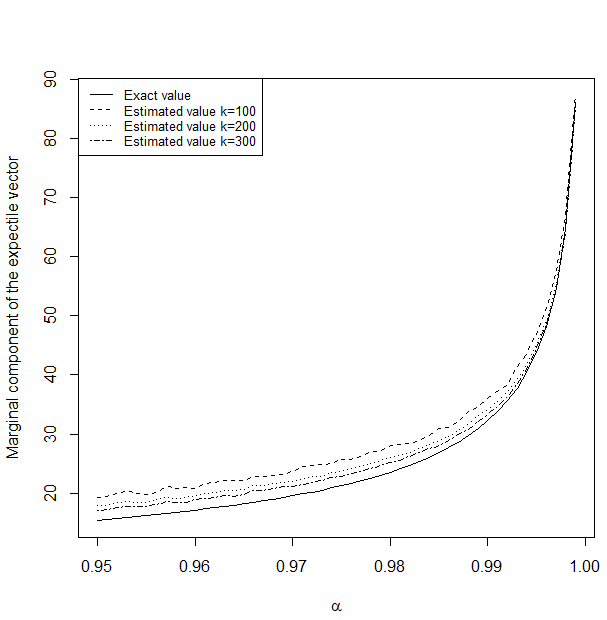} 
	\includegraphics[width=8cm]{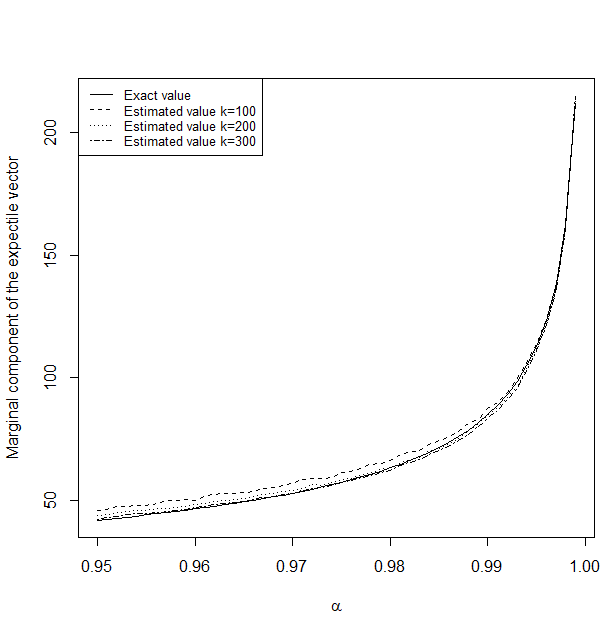}
	\caption{Convergence of $\hat{\mathbf{e}}^{\bot}_{\alpha}(\mathbf{X})$ (asymptotic independence case). On the left, the first coordinate of ${\mathbf{e}}_{\alpha}(\mathbf{X})$ and $\hat{\mathbf{e}}^{\bot}_{\alpha}(\mathbf{X})$ for various values of $k=k(n)$ are plotted. The right figure concerns the second coordinate.}
	\label{BurrI}
\end{figure}
\begin{figure}[H]
	\center
	\includegraphics[width=8cm]{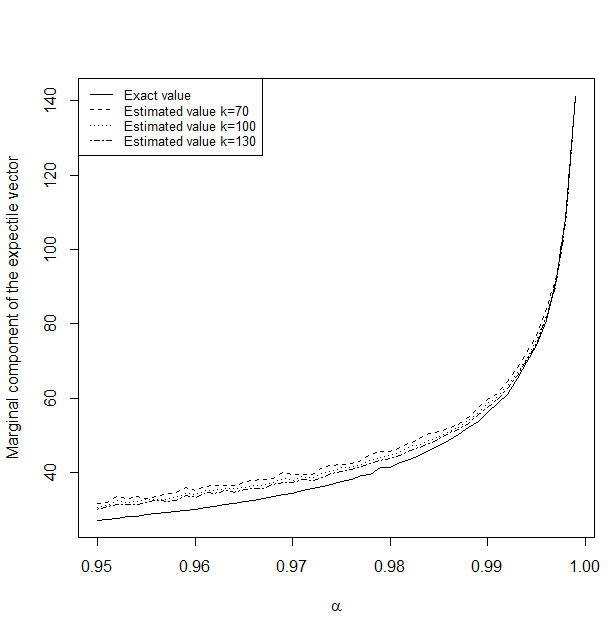} 
	\includegraphics[width=8cm]{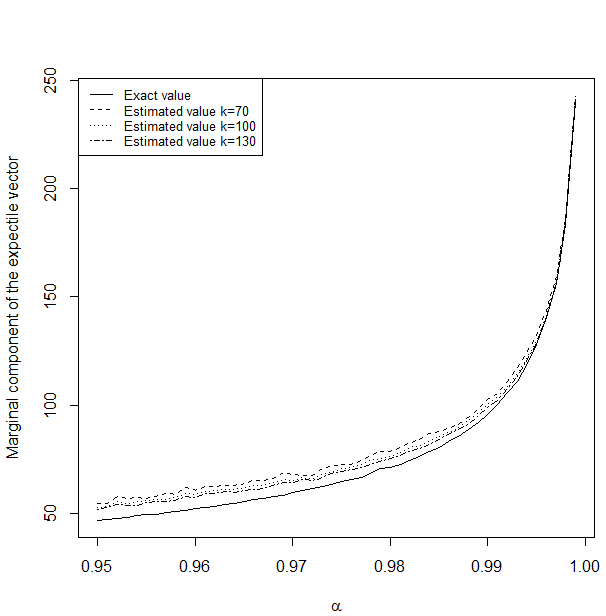}
	\caption{Convergence of $\hat{\mathbf{e}}^{+}_{\alpha}(\mathbf{X})$ (comonotonic case). On the left, the first coordinate of ${\mathbf{e}}_{\alpha}(\mathbf{X})$ and $\hat{\mathbf{e}}^{\bot}_{\alpha}(\mathbf{X})$ for various values of $k=k(n)$ are plotted. The right figure concerns the second coordinate.}
	\label{BurrC}
\end{figure}
\subsubsection{Student distributions}
In order to illustrate the convergence of the two estimators for other distributions nature, we close this subsection by a Student model. We consider a risk vector $(X_1,\ldots,X_d)$ such that $X_i=a_iT_i$ for all $i\in\{1,\ldots,d\}$ and $(T_i)_i$ are identically distributed following a t-distribution of parameter $z$. Using L'Hôpital's rule, the tails are equivalent since
 $$\underset{x\rightarrow+\infty}{\lim}\frac{\bar{F}_{X_i}(x)}{\bar{F}_{X_1}(x)}=\underset{x\rightarrow+\infty}{\lim}\frac{\bar{F}_{T_1}(x/a_i)}{\bar{F}_{X_1}(x)}=\underset{x\rightarrow+\infty}{\lim}\frac{a_1f_{T_1}(x/a_i)}{a_if_{T_1}(x/a_1)}=\left(\frac{a_i}{a_1}\right)^z=c_i, \forall i\in\{2,\ldots,d\}\/.$$ 
 The marginal tails are all $RV_{-(z+1)}(+\infty)$. For the Student comonotonic model $\mu_{i,j}=\frac{a_i}{a_j}x_j$. \\  
 For the numerical illustration the parameters are $a_i=2^{i-1}$ for $i=1,\ldots,d$ and $z=2$. In the case of the independence $(T_i)_i$ are supposed independent, and they are comonotonic in the comonotonic case.\\ 
Figures \ref{TI} and \ref{TC} present an illustration of the obtained results in the two cases.
\begin{figure}[H]
	\center
	\includegraphics[width=8cm]{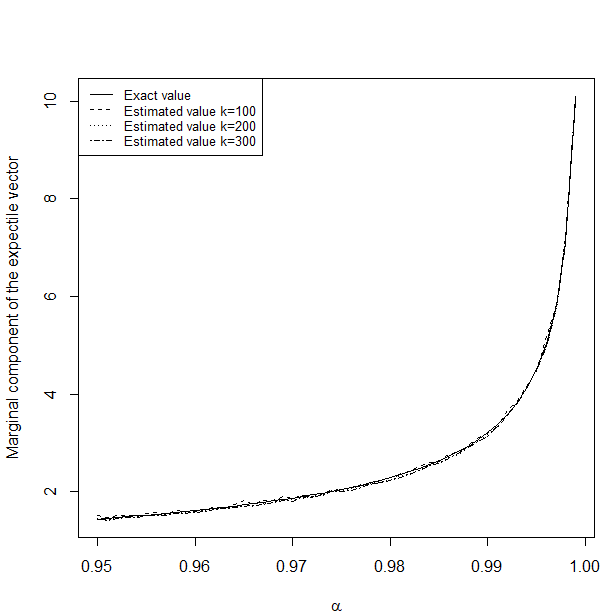} 
	\includegraphics[width=8cm]{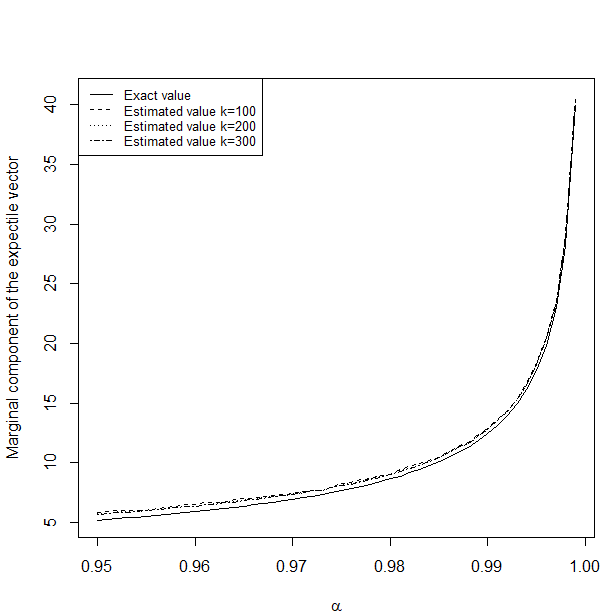}
	\caption{Convergence of $\hat{\mathbf{e}}^{\bot}_{\alpha}(\mathbf{X})$ (asymptotic independence case). On the left, the first coordinate of ${\mathbf{e}}_{\alpha}(\mathbf{X})$ and $\hat{\mathbf{e}}^{\bot}_{\alpha}(\mathbf{X})$ for various values of $k=k(n)$ are plotted. The right figure concerns the second coordinate.}
	\label{TI}
\end{figure}
\begin{figure}[H]
	\center
	\includegraphics[width=8cm]{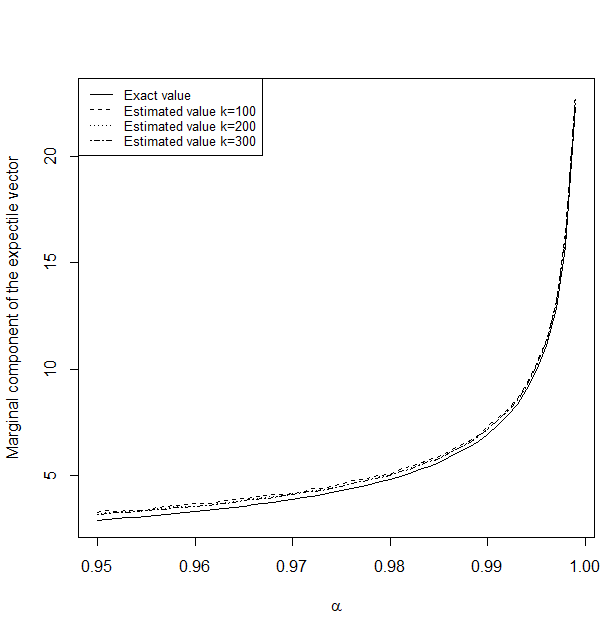} 
	\includegraphics[width=8cm]{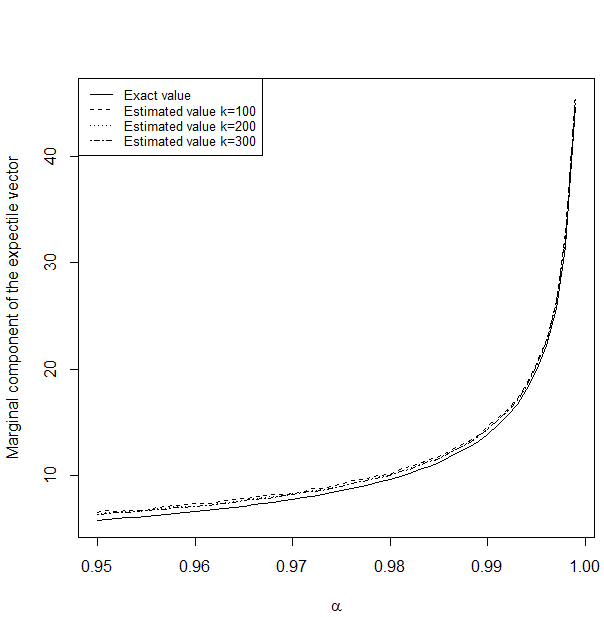}
	\caption{Convergence of $\hat{\mathbf{e}}^{+}_{\alpha}(\mathbf{X})$ (comonotonic case). On the left, the first coordinate of ${\mathbf{e}}_{\alpha}(\mathbf{X})$ and $\hat{\mathbf{e}}^{\bot}_{\alpha}(\mathbf{X})$ for various values of $k=k(n)$ are plotted. The right figure concerns the second coordinate.}
	\label{TC}
\end{figure}

For the three Fréchet models, Pareto, Weibull, and Student, the different illustrations show that the convergence is better for values of $\alpha$ close to 1. This is natural since we are approaching the extreme level and therefore the estimate value converges towards the theoretical value. The convergence seems to be stable for values of $k(n)$ in the convergence zone. When $\alpha$ moves away from 1, the difference with the theoretical value is apparently a function of the marginal risk level represented by the coefficients of tails' equivalence $c_i$.

 \section*{Conclusion}
We have studied properties of extreme multivariate expectiles in a regular variations framework. We have seen that the asymptotic behavior of expectiles vectors strongly depends on the marginal tails behavior and on the nature of the asymptotic dependence. The main conclusion of this analysis, is that the equivalence of marginal tails leads to equivalence of the extreme expectile components.\\
 
 The statistical estimation of the integrals of the tail dependence functions would allow to construct estimators of the extreme expectile vectors. This paper's estimations are limited to the cases of asymptotic independence and comonotonicity which do not require the estimation of the tail dependence functions. The asymptotic normality of the estimators proposed in the last section of this paper requires a careful technical analysis which is not considered in this paper.\\
   
 A natural perspective of this work, is to study the asymptotic behavior of $\Sigma$-expectiles in the case of equivalent tails of marginal distributions in the domains of attraction of Weibull and Gumbel. The Gumbel's domain contains most of the usual distributions, especially the family of Weibull tail-distributions, which makes the analysis of its case an interesting task.
 
\section*{Acknowledgment}
 We are grateful to the editor and to the reviewers. We thank them for their valuable comments which helped to improve the quality of our manuscript.
	\bibliographystyle{plain}
	\bibliography{bibAsyMultExp} 
	\newpage
\appendix
\section{Proofs}
\subsection{Lemma \ref{Lemme+inf}}\label{pro1}
 \begin{proof}
 	We give some details on  the proof for the first item, the second one may be obtained in the same way.\\
 	Under H1 and H2, for all $i\in\{1,\ldots,d\}$ $$\bar{F}_{X_i}\in \mbox{RV}_{-\theta}(+\infty)\/,$$ 
 	there exists for all $i$ a positive measurable function $L_i\in \mbox{RV}_{0}(+\infty)$ such that
 	$$\bar{F}_{X_i}(x)=x^{-\theta}L_i(x), \forall x>0\/,$$
 	then for all $(i,j)\in\{1,\ldots,d\}^2$ and all $t,s>0$
 	\begin{equation}\label{Lemme+infEq1}
 	\frac{s\bar{F}_{X_i}(s)}{t\bar{F}_{X_j}(t)}=\left(\frac{s}{t}\right)^{-\theta+1}\frac{L_i(s)}{L_j(t)}=\left(\frac{s}{t}\right)^{-\theta+1}\frac{L_i(s)}{L_i(t)}\frac{L_i(t)}{L_j(t)}\/, 
 	\end{equation}
 	and under H2
 	\begin{equation}\label{Lemme+infEq2}
 	\underset{x\uparrow+\infty}{\lim}\frac{L_i(x)}{L_j(x)}=\frac{c_i}{c_j}\/.
 	\end{equation}
 	Using  Karamata's representation for slowly varying functions (Theorem \ref{Karamata-Rep}), there exist a constant $c>0$, a positive measurable function $c(\cdot)$ with $\underset{x\uparrow+\infty}{\lim}c(x)=c>0$, 
 	such that  $\forall \epsilon>0$,  $\exists ~t_0$ such that $\forall~t>t_0$ 
 	
 	$$\frac{L_i(s)}{L_i(t)}\leq \left(\frac{s}{t}\right)^{\epsilon}\frac{c(s)}{c(t)}\/.$$
 	Taking $0<\epsilon<\theta-1$, we conclude
 	$$\underset{t\uparrow+\infty}{\lim}\frac{s\bar{F}_{X_i}(s)}{t\bar{F}_{X_j}(t)}=0, ~\forall (i,j)\in\{1,\ldots,d\}^2\/.$$	 
 \end{proof}
\subsection{Proposition \ref{PropLimits1V1}}\label{pro2}
\begin{proof} We start by proving that
	 $$\underset{\alpha\longrightarrow 1}{\overline{\lim}}\frac{1-\alpha}{\bar{F}_{X_i}(\mathbf{e}^i_\alpha(\mathbf{X}))}<+\infty, \forall i\in\{2,\ldots,d\}\/.
	$$ 
	Using H2, it is sufficient to show that
	$$
	\underset{\alpha\longrightarrow 1}{\overline{\lim}}\frac{1-\alpha}{\bar{F}_{X_1}(\mathbf{e}^1_\alpha(\mathbf{X}))}<+\infty\/. 
	$$ 
	Assume that $
	\underset{\alpha\longrightarrow 1}{\overline{\lim}}\frac{1-\alpha}{\bar{F}_{X_1}(\mathbf{e}^1_\alpha(\mathbf{X}))}=+\infty$, we shall prove that, in that case, (\ref{S-functionsL2}) cannot be satisfied. Taking if necessary a subsequence ($\alpha_n\rightarrow 1$), we may assume that $
	\underset{\alpha\longrightarrow 1}{\lim}\frac{1-\alpha}{\bar{F}_{X_1}(\mathbf{e}^1_\alpha(\mathbf{X}))}=+\infty$.\\
	We have 
	\begin{align*}
	\frac{l^\alpha_{X_1}(x_1)}{(1-\alpha)x_1}&=\left(\frac{\alpha\mathbb{E}[(X_1-x_1)_+]-(1-\alpha)\mathbb{E}[(X_1-x_1)_-]}{(1-\alpha)x_1}\right)\\
	&=\left((2\alpha-1)\frac{\mathbb{E}[(X_1-x_1)_+]}{(1-\alpha)x_1}-\frac{(1-\alpha)(x_1-\mathbb{E}[X_1])}{(1-\alpha)x_1}\right)\\
	&=\left((2\alpha-1)\frac{\mathbb{E}[(X_1-x_1)_+]}{x_1\bar{F}_{X_1}(x_1)}\frac{\bar{F}_{X_1}(x_1)}{1-\alpha}-1+\frac{\mathbb{E}[X_1]}{x_1}\right)\overset{\alpha\uparrow 1}{\longrightarrow}-1\ \mbox{recall (\ref{EqFrechetSL}) }.
	\end{align*} 
	Furthermore, for all $i\in\{2,\ldots,d\}$
	\begin{align*}
	\frac{l^\alpha_{X_i,X_1}(x_i,x_1)}{(1-\alpha)x_1}&=\left(\frac{\alpha\mathbb{E}[(X_i-x_i)_+ 1\!\!1_{\{X_1>x_1\}}]-(1-\alpha)\mathbb{E}[(X_i-x_i)_- 1\!\!1_{\{X_1<x_1\}}]}{(1-\alpha)x_1}\right)\\
	&=\left(\frac{\alpha\mathbb{E}[(X_i-x_i)_+ 1\!\!1_{\{X_1>x_1\}}]}{(1-\alpha)x_1} - \frac{\mathbb{E}[(X_i-x_i)_+ 1\!\!1_{\{X_1<x_1\}}]}{x_1} -\frac{x_i\mathbb{P}(X_1<x_1)}{x_1}+ \frac{\mathbb{E}[X_i1\!\!1_{\{X_1<x_1\}}]}{x_1}\right)\/.
	\end{align*}
	On one side,
	\begin{align*}
	\frac{\mathbb{E}[(X_i-x_i)_+ 1\!\!1_{\{X_1>x_1\}}]}{(1-\alpha)x_1}&\leq\frac{\mathbb{E}[(X_i-x_i)_+ ]}{(1-\alpha)x_1}=\frac{\bar{F}_{X_1}(x_1)}{1-\alpha}\frac{\mathbb{E}[(X_i-x_i)_+]}{x_i\bar{F}_{X_i}(x_i)}\frac{x_i\bar{F}_{X_i}(x_i)}{x_1\bar{F}_{X_1}(x_1)}, \forall i\in\{2,\ldots,d\}\/.
	\end{align*}
	So that, Lemma \ref{Lemme+inf} implies
	\begin{equation}\label{P1Eq1}
	\underset{\alpha\longrightarrow 1}{\lim}\frac{\mathbb{E}[(X_i-x_i)_+ 1\!\!1_{\{X_1>x_1\}}]}{(1-\alpha)x_1}=0, \forall k\in J_C^1\cup J_\infty^1\/.
	\end{equation}
	
	Let  $i\in J_0^1$, taking if necessary a subsequence, we may assume that $\frac{x_i}{x_1}\rightarrow 0$.
	\begin{equation}\label{3.2.1-Eq1}
	\frac{\mathbb{E}[(X_i-x_i)_+ 1\!\!1_{\{X_1>x_1\}}]}{(1-\alpha)x_1}=\frac{\displaystyle\int_{x_i}^{x_1}\mathbb{P}\left(X_i>t,X_1>x_1\right)dt}{(1-\alpha)x_1}+\frac{\displaystyle\int_{x_1}^{+\infty}\mathbb{P}\left(X_i>t,X_1>x_1\right)dt}{(1-\alpha)x_1}\/.
	\end{equation}
	Now,
	$$\frac{\displaystyle\int_{x_i}^{x_1}\mathbb{P}\left(X_i>t,X_1>x_1\right)dt}{(1-\alpha)x_1}\leq\frac{\displaystyle\int_{x_i}^{x_1}\mathbb{P}\left(X_1>x_1\right)dt}{(1-\alpha)x_1}=\frac{\bar{F}_{X_1}(x_1)}{1-\alpha}\left(1-\frac{x_i}{x_1}\right)\/.$$
	Thus,
	\begin{equation}\label{3.2.1-Eq2}
	\underset{\alpha\longrightarrow 1}{\lim}\frac{\displaystyle\int_{x_i}^{x_1}\mathbb{P}\left(X_i>t,X_1>x_1\right)dt}{(1-\alpha)x_1}=0\/.
	\end{equation}
	Consider the second term of (\ref{3.2.1-Eq1})
	$$\frac{\displaystyle\int_{x_1}^{+\infty}\mathbb{P}\left(X_i>t,X_1>x_1\right)dt}{(1-\alpha)x_1}\leq\frac{\displaystyle\int_{x_1}^{+\infty}\mathbb{P}\left(X_i>t\right)dt}{(1-\alpha)x_1}\/,$$
	Karamata's Theorem (Theorem \ref{Th-Karamata}) gives
	$$\frac{\displaystyle\int_{x_1}^{+\infty}\mathbb{P}\left(X_k>t\right)dt}{(1-\alpha)x_1}\overset{\alpha\uparrow1}{\sim}\frac{1}{\theta-1}\frac{\bar{F}_{X_k}(x_1)}{1-\alpha}\/,$$
	which leads to
	\begin{equation}\label{3.2.1-Eq3}
	\underset{\alpha\longrightarrow 1}{\lim}\frac{\displaystyle\int_{x_1}^{+\infty}\mathbb{P}\left(X_i>t,X_1>x_1\right)dt}{(1-\alpha)x_1}=0\/.
	\end{equation}
	Finally, we get
	\begin{equation}\label{P1Eq3}
	\underset{\alpha\longrightarrow 1}{\lim}\frac{\mathbb{E}[(X_i-x_i)_+ 1\!\!1_{\{X_1>x_1\}}]}{(1-\alpha)x_1}=0, \forall i\in J_0^1\/.
	\end{equation} 
	We have  shown that
	$$\underset{\alpha\longrightarrow 1}{\lim}\frac{\mathbb{E}[(X_k-x_k)_+ 1\!\!1_{\{X_1>x_1\}}]}{(1-\alpha)x_1}=0, \forall k\in\{2,\ldots,d\}\/,$$
	so, the first equation of optimality system (\ref{S-functionsL2}) implies that  
	$$-\lim_{\alpha\rightarrow 1}\left(\sum_{k\in J_0^1\setminus J_\infty^1}^{}\frac{l^\alpha_{X_k,X_1}(x_k,x_1)}{(1-\alpha)x_1}+\sum_{k\in J_C^1}^{}\frac{l^\alpha_{X_k,X_1}(x_k,x_1)}{(1-\alpha)x_1}+\sum_{k\in J_\infty^1}^{}\frac{l^\alpha_{X_k,X_1}(x_k,x_1)}{(1-\alpha)x_1}\right)=\underset{\alpha\longrightarrow 1}{\lim}\sum_{k=2}^{d}\frac{x_k}{x_1}=-1\/,$$
	this is absurd since the $x_k$'s are non negative, and consequently  
	$$\underset{\alpha\longrightarrow1}{\varlimsup}\frac{1-\alpha}{\bar{F}_{X_1}(x_1)}<+\infty\/.$$
Now, we prove that the components of the extreme multivariate expectile satisfy also
$$
0<\underset{\alpha\longrightarrow 1}{\underline{\lim}}\frac{1-\alpha}{\bar{F}_{X_i}(\mathbf{e}^i_\alpha(\mathbf{X}))}, \forall i\in\{2,\ldots,d\}\/.
$$ 	
 Using H2, it is sufficient to show that 	
	$$
	0<\underset{\alpha\longrightarrow 1}{\underline{\lim}}\frac{1-\alpha}{\bar{F}_{X_1}(\mathbf{e}^1_\alpha(\mathbf{X}))}\/. 
	$$ 
	Let us assume that $
	\underset{\alpha\longrightarrow 1}{\underline{\lim}}\frac{1-\alpha}{\bar{F}_{X_1}(\mathbf{e}^1_\alpha(\mathbf{X}))}=0$, we shall see that in that case, (\ref{S-functionsL2}) cannot be satisfied. Taking if necessary a convergent subsequence, we may assume that $
	\underset{\alpha\longrightarrow 1}{\lim}\frac{1-\alpha}{\bar{F}_{X_1}(\mathbf{e}^1_\alpha(\mathbf{X}))}=0$. In this case,  \begin{equation*}\label{Case3-Eq1}
	\frac{l^\alpha_{X_1}(x_1)}{x_1\bar{F}_{X_1}(x_1)}=\left((2\alpha-1)\frac{\mathbb{E}[(X_1-x_1)_+]}{x_1\bar{F}_{X_1}(x_1)}-\frac{1-\alpha}{\bar{F}_1(x_1)}(1-\frac{\mathbb{E}[X_1]}{x_1})\right)\overset{\alpha\uparrow 1}{\longrightarrow}\frac{1}{\theta-1}>0\/.
	\end{equation*} 
	On another side, let  $i\in J^1_\infty$, taking if necessary a subsequence, we may assume that $x_1=o(x_i)$. Lemma \ref{Lemme+inf} and Proposition \ref{PropLimits1V1} give:
	$$	\frac{1-\alpha}{\bar{F}_1(x_1)}\frac{x_i}{x_1}= \frac{1-\alpha}{\overline{F}_{X_i}(x_i)} \cdot \frac{x_i\overline{F}_{X_i}(x_i)}{x_1\overline{F}_{X_1}(x_1)} \longrightarrow 0 \ \mbox{as} \ \alpha \rightarrow 1\/.$$
	Moreover,
	$$
	\frac{\mathbb{E}[(X_i-x_i)_+ 1\!\!1_{\{X_1>x_1\}}]}{x_1\bar{F}_{X_1}(x_1)}\leq\frac{\mathbb{E}[(X_i-x_i)_+ ]}{x_1\bar{F}_{X_1}(x_1)}=\frac{\mathbb{E}[(X_i-x_i)_+]}{x_i\bar{F}_{X_i}(x_i)}\frac{x_i\bar{F}_{X_i}(x_i)}{x_1\bar{F}_{X_1}(x_1)}\/, 
	$$
	We deduce
	\begin{align*}
	\frac{l^\alpha_{X_i,X_1}(x_i,x_1)}{x_1\bar{F}_{X_1}(x_1)}&=\left(\frac{\alpha\mathbb{E}[(X_i-x_i)_+ 1\!\!1_{\{X_1>x_1\}}]-(1-\alpha)\mathbb{E}[(X_i-x_i)_- 1\!\!1_{\{X_1<x_1\}}]}{x_1\bar{F}_{X_1}(x_1)}\right)\\
	&\longrightarrow 0, ~\forall i\in J_\infty^1\/.
	\end{align*}
	
	Going through the limit ($\alpha\longrightarrow1$) in the first equation of the optimality system (\ref{S-functionsL2}) divided by $x_1\bar{F}_{X_1}(x_1)$, leads to
	\begin{equation*}\label{Case3-Eq2}
	\underset{\alpha\longrightarrow 1}{\lim}\sum_{k\in J^1_0\cup J^1_C \setminus J_\infty^1}^{}\frac{l^\alpha_{X_k,X_1}(x_k,x_1)}{x_1\bar{F}_{X_1}(x_1)}=-\frac{1}{\theta-1}\/,
	\end{equation*}
	which is absurd because 
	\begin{align*}
	\underset{\alpha\longrightarrow 1}{\lim}\sum_{k\in J^1_0\cup J^1_C\setminus J_\infty^1}^{}\frac{l^\alpha_{X_k,X_1}(x_k,x_1)}{x_1\bar{F}_{X_1}(x_1)}&=\underset{\alpha\longrightarrow 1}{\lim}\sum_{k\in J^1_0\cup J^1_C\setminus J_\infty^1}^{}\left(\frac{\alpha\mathbb{E}[(X_k-x_k)_+ 1\!\!1_{\{X_1>x_1\}}]-(1-\alpha)\mathbb{E}[(X_k-x_k)_- 1\!\!1_{\{X_1<x_1\}}]}{x_1\bar{F}_{X_1}(x_1)}\right)\\
	&=\underset{\alpha\longrightarrow 1}{\lim}\sum_{k\in J^1_0\cup J^1_C\setminus J_\infty^1}^{}\left(\frac{\mathbb{E}[(X_k-x_k)_+ 1\!\!1_{\{X_1>x_1\}}]}{x_1\bar{F}_{X_1}(x_1)}-\frac{1-\alpha}{\bar{F}_{X_1}(x_1)}\frac{x_k}{x_1}\right)\\
	&=\underset{\alpha\longrightarrow 1}{\lim}\sum_{k\in J^1_0\cup J^1_C\setminus J_\infty^1}^{}\left(\frac{\mathbb{E}[(X_k-x_k)_+ 1\!\!1_{\{X_1>x_1\}}]}{x_1\bar{F}_{X_1}(x_1)}\right)\geq0\/.
	\end{align*}
	We can finally conclude that
	$$\underset{\alpha\longrightarrow1}{\varliminf}\frac{1-\alpha}{\bar{F}_{X_1}(x_1)}>0\/.$$
\end{proof}
\subsection{Lemma \ref{L1-Dom-Frechet}}\label{pro4}

 \begin{proof}[Proof]
	Taking if necessary a convergent subsequence $(\alpha_n)_{n\in\mathbb{N}}$ $\alpha_n\longrightarrow 1$, we consider that the limits  $\underset{\alpha\rightarrow1}{\lim}\frac{x_i}{x_1}=\beta_i$ exist.\\ 
	Using the notation $J_{C}=\{i\in\{2,\ldots,d\}|0<\beta_i<+\infty\}$, for all $i\in J_{C}$ 
	$$\underset{\alpha\uparrow1}{\lim}\frac{\mathbb{E}[(X_i-x_i)_+1\!\!1_{\{X_1>x_1\}}]}{x_1\bar{F}_{X_1}(x_1)}=\underset{\alpha\uparrow1}{\lim}\intevariable{\beta_i}{+\infty}{\frac{\mathbb{P}(X_i>tx_1,X_1>x_1)}{\mathbb{P}(X_1>x_1)}dt}\/,$$
	because $$\frac{\mathbb{P}(X_i>tx_1,X_1>x_1)}{\mathbb{P}(X_1>x_1)}=\mathbb{P}(X_i>tx_1|X_1>x_1)\leq1\/.$$  
	On another hand,
	$$\frac{\mathbb{P}(X_i>tx_1,X_1>x_1)}{\mathbb{P}(X_1>x_1)}\leq \min\{1, \frac{\mathbb{P}(X_i>tx_1)}{\mathbb{P}(X_1>x_1)}\}\/,$$
	and $$\frac{\mathbb{P}(X_i>tx_1)}{\mathbb{P}(X_1>x_1)}=\frac{\bar{F}_{X_i}(tx_1)}{\bar{F}_{X_1}(tx_1)}\frac{\bar{F}_{X_1}(tx_1)}{\bar{F}_{X_1}(x_1)}\/,$$ 
	then, using  $\underset{\alpha\uparrow1}{\lim}\frac{\bar{F}_{X_i}(tx_1)}{\bar{F}_{X_1}(tx_1)}=0$ and the Potter's bounds (\ref{Potter}) associated to $\bar{F}_{X_1}$, we deduce that for all $\epsilon_1>0$  and $0<\epsilon_2<1$, there exists $x^0_1(\epsilon_1,\epsilon_2)$ such that for $x_1\geq\frac{x^0_1(\epsilon_1,\epsilon_2)}{\min\{1,\beta_i\}}$ 
	$$\frac{\mathbb{P}(X_i>tx_1)}{\mathbb{P}(X_1>x_1)}\leq\epsilon_1(1+\epsilon_2)\max\left(t^{-\theta+\epsilon_2},t^{-\theta-\epsilon_2}\right)\/.$$
	And the application of the Dominated Convergence Theorem leads to 
	$$\underset{\alpha\uparrow1}{\lim}\frac{\mathbb{E}[(X_i-x_i)_+1\!\!1_{\{X_1>x_1\}}]}{x_1\bar{F}_{X_1}(x_1)}=\intevariable{\beta_i}{+\infty}{\underset{\alpha\uparrow1}{\lim}\frac{\mathbb{P}(X_i>tx_1,X_1>x_1)}{\mathbb{P}(X_1>x_1)}dt}=0,~\forall~i\in J_{C}\/. $$
	We denote by $J_{\infty}$ the set $J_{\infty}=\{i\in\{2,\ldots,d\}|\beta_i=+\infty\}$. So, for all  $i\in J_{\infty}$,  $x_1=o(x_i)$ and
	\begin{align*}
	\frac{\mathbb{E}[(X_i-x_i)_+1\!\!1_{\{X_1>x_1\}}]}{x_1\bar{F}_{X_1}(x_1)}&=\intevariable{x_i}{+\infty}{\frac{\mathbb{P}(X_i>t,X_1>x_1)}{x_1\mathbb{P}(X_1>x_1)}dt}\\
	&\leq\intevariable{x_1}{+\infty}{\frac{\mathbb{P}(X_i>t,X_1>x_1)}{x_1\mathbb{P}(X_1>x_1)}dt}=\intevariable{1}{+\infty}{\frac{\mathbb{P}(X_i>tx_1,X_1>x_1)}{\mathbb{P}(X_1>x_1)}dt}\/.
	\end{align*}
	In the same way as in the previous case, and using the Potter's bounds, we show that
	$$\underset{\alpha\uparrow1}{\lim}\intevariable{1}{+\infty}{\frac{\mathbb{P}(X_i>tx_1,X_1>x_1)}{\mathbb{P}(X_1>x_1)}dt}=\intevariable{1}{+\infty}{\underset{\alpha\uparrow1}{\lim}\frac{\mathbb{P}(X_i>tx_1,X_1>x_1)}{\mathbb{P}(X_1>x_1)}dt}=0\/,$$
	from which we deduce that
	$$\underset{\alpha\uparrow1}{\lim}\frac{\mathbb{E}[(X_i-x_i)_+1\!\!1_{\{X_1>x_1\}}]}{x_1\bar{F}_{X_1}(x_1)}=0,~\forall~i\in J_{\infty}\/.$$
	Let $J_{0}$ be the set $J_{0}=\{i\in\{2,\ldots,d\}|\beta_i=0\}$. For all $i\in J_{0}$ we have $x_i=o(x_1)$, then 
	$$\underset{\alpha\uparrow1}{\lim}\frac{\mathbb{E}[(X_i-x_i)_+1\!\!1_{\{X_1>x_1\}}]}{x_1\bar{F}_{X_1}(x_1)}=\underset{\alpha\uparrow1}{\lim}\intevariable{x_i}{+\infty}{\frac{\mathbb{P}(X_i>t,X_1>x_1)}{x_1\mathbb{P}(X_1>x_1)}dt}=\underset{\alpha\uparrow1}{\lim}\intevariable{0}{+\infty}{\frac{\mathbb{P}(X_i>tx_1,X_1>x_1)}{\mathbb{P}(X_1>x_1)}dt}\/,$$
	because $\underset{\alpha\uparrow1}{\lim}\frac{x_i}{x_1}=0$ and $\frac{\mathbb{P}(X_i>tx_1,X_1>x_1)}{\mathbb{P}(X_1>x_1)}\leq1$. \\In addition, for all $\epsilon>0$, we have $$\underset{\alpha\uparrow1}{\lim}\intevariable{\epsilon}{+\infty}{\frac{\mathbb{P}(X_i>tx_1,X_1>x_1)}{\mathbb{P}(X_1>x_1)}dt}=0\/,$$
	because the Dominated Convergence Theorem is applicable using the Potter's bounds, and 
	$\underset{\alpha\uparrow1}{\lim}\frac{\mathbb{P}(X_i>tx_1,X_1>x_1)}{\mathbb{P}(X_1>x_1)}=0$ for all $t>0$ since $c_i=0$.\\
	Let $\kappa>0$, $\forall \epsilon>0$ $\exists \alpha_0$ such that $\forall \alpha>\alpha_0$
	$$\intevariable{\epsilon}{+\infty}{\frac{\mathbb{P}(X_i>tx_1,X_1>x_1)}{\mathbb{P}(X_1>x_1)}dt}<\kappa\/,$$
	then
	$$\intevariable{0}{+\infty}{\frac{\mathbb{P}(X_i>tx_1,X_1>x_1)}{\mathbb{P}(X_1>x_1)}dt}=\intevariable{0}{\epsilon}{\frac{\mathbb{P}(X_i>tx_1,X_1>x_1)}{\mathbb{P}(X_1>x_1)}dt}+\intevariable{\epsilon}{+\infty}{\frac{\mathbb{P}(X_i>tx_1,X_1>x_1)}{\mathbb{P}(X_1>x_1)}dt}<\epsilon+\kappa\/,$$
	we deduce that 
	$$\underset{\alpha\uparrow1}{\lim}\intevariable{0}{+\infty}{\frac{\mathbb{P}(X_i>tx_1,X_1>x_1)}{\mathbb{P}(X_1>x_1)}dt}=0\/,$$
	so, \begin{equation*}
	\underset{\alpha\uparrow1}{\lim}\frac{\mathbb{E}[(X_i-x_i)_+1\!\!1_{\{X_1>x_1\}}]}{x_1\bar{F}_{X_1}(x_1)}=0,~~\forall i\in J_{0}\/.
	\end{equation*}
	We have therefore shown that
	\begin{equation*}
	\underset{\alpha\uparrow1}{\lim}\frac{\mathbb{E}[(X_i-x_i)_+1\!\!1_{\{X_1>x_1\}}]}{x_1\bar{F}_{X_1}(x_1)}=0,~~\forall i\in \{2,\ldots,d\}\/.
	\end{equation*} 
\end{proof}	 
\subsection{Lemma \ref{L2-Dom-Frechet}}\label{pro5}
 \begin{proof}
	We suppose that $
	\underset{\alpha\uparrow1}{\varlimsup}\frac{1-\alpha}{\bar{F}_{X_1}(\mathbf{e}_\alpha^1(\mathbf{X}))}=+\infty
	$. Taking if necessary a convergent subsequence  $(\alpha_n)_{n\in\mathbb{N}}$ with  $\alpha_n\longrightarrow 1$, we consider that the limits $\underset{\alpha\rightarrow1}{\lim}\frac{x_i}{x_1}=\beta_i$ exist and that $\underset{\alpha\uparrow1}{\lim}\frac{1-\alpha}{\bar{F}_{X_1}(x_1)}=+\infty$.\\ 
	We use the notations $J_{C}=\{i\in\{2,\ldots,d\}|~0<\beta_i<+\infty\}$,  $J_{0}=\{i\in\{2,\ldots,d\}|~\beta_i=0\}$, and $J_{\infty}=\{i\in\{2,\ldots,d\}|~\beta_i=+\infty\}$.\\
	The first equation of the optimality system (\ref{eq1}) divided by $x_1\bar{F}_{X_1}(x_1)$ can be written
	\begin{align*}
	\frac{(2\alpha-1)\mathbb{E}[(X_1-x_1)_+]}{x_1\bar{F}_{X_1}(x_1)}+\sum_{i=2}^{d}\frac{\alpha\mathbb{E}[(X_i-x_i)_+1\!\!1_{\{X_1>x_1\}}]}{x_1\bar{F}_{X_1}(x_1)}&=\frac{1-\alpha}{\bar{F}_{X_1}(x_1)}\left(\frac{x_1-\mathbb{E}[X_1]}{x_1}\right)\\
	&~~+\frac{1-\alpha}{\bar{F}_{X_1}(x_1)}\left(\sum_{i=2}^{d}\frac{\mathbb{E}[(X_i-x_i)_-1\!\!1_{\{X_1<x_1\}}]}{x_1}\right)\/.
	\end{align*}
	By (\ref{EqFrechetSL}) $$\underset{\alpha\uparrow1}{\lim}\frac{(2\alpha-1)\mathbb{E}[(X_1-x_1)_+]}{x_1\bar{F}_{X_1}(x_1)}=\frac{1}{\theta-1}\/,$$
	and by Lemma \ref{L1-Dom-Frechet}
	$$\underset{\alpha\uparrow1}{\lim}\sum_{i=2}^{d}\frac{\alpha\mathbb{E}[(X_i-x_i)_+1\!\!1_{\{X_1>x_1\}}]}{x_1\bar{F}_{X_1}(x_1)}=0\/,$$
	so, going through the limit ($\alpha\rightarrow1$) in the previous equation leads to
	$$\underset{\alpha\uparrow1}{\lim}\frac{1-\alpha}{\bar{F}_{X_1}(x_1)}\left(\frac{x_1-\mathbb{E}[X_1]}{x_1}+\sum_{i=2}^{d}\frac{\mathbb{E}[(X_i-x_i)_-1\!\!1_{\{X_1<x_1\}}]}{x_1}\right)=\frac{1}{\theta-1}\/,$$
	nevertheless, 
	$$\underset{\alpha\uparrow1}{\lim}\frac{1-\alpha}{\bar{F}_{X_1}(x_1)}\left(\frac{x_1-\mathbb{E}[X_1]}{x_1}+\sum_{i=2}^{d}\frac{\mathbb{E}[(X_i-x_i)_-1\!\!1_{\{X_1<x_1\}}]}{x_1}\right)=\underset{\alpha\uparrow1}{\lim}\frac{1-\alpha}{\bar{F}_{X_1}(x_1)}\left(1+\sum_{i=2}^{d}\frac{x_i}{x_1}\right)=+\infty\/.$$
	From this contradiction, we deduce that the case 
	$\underset{\alpha\uparrow1}{\lim}\frac{1-\alpha}{\bar{F}_{X_1}(x_1)}=+\infty$ is absurd.\\

	Now, we suppose that $
	\underset{\alpha\uparrow1}{\varliminf}\frac{1-\alpha}{\bar{F}_{X_1}(\mathbf{e}_\alpha^1(\mathbf{X}))}=0
	$. Taking if necessary a subsequence $(\alpha_n)_{n\in\mathbb{N}}$ with $\alpha_n\longrightarrow 1$, we consider that the limits  $\underset{\alpha\rightarrow1}{\lim}\frac{x_i}{x_1}=\beta_i$ exist and that $\underset{\alpha\uparrow1}{\lim}\frac{1-\alpha}{\bar{F}_{X_1}(x_1)}=0$.\\ 
	We denote $J_{C}=\{i\in\{2,\ldots,d\}|~0<\beta_i<+\infty\}$,  $J_{0}=\{i\in\{2,\ldots,d\}|~\beta_i=0\}$, and $J_{\infty}=\{i\in\{2,\ldots,d\}|~\beta_i=+\infty\}$.\\
	Going through the limit $(\alpha\rightarrow 1)$ in the first equation of System \ref{eq1} divided by $x_1\bar{F}_{X_1}(x_1)$, and using Lemma \ref{L1-Dom-Frechet} and Equation \ref{EqFrechetSL}, leads to 
	\begin{equation}\label{DAsy-eq2}
	\underset{\alpha\uparrow1}{\lim}\left(\frac{1-\alpha}{\bar{F}_{X_1}(x_1)}\sum_{i\in J_\infty}^{}\frac{x_i}{x_1}\right)=\frac{1}{\theta-1}\/.
	\end{equation}
	If $J_\infty\neq\emptyset$, so, there exists $i\in\{2,\ldots,d\}$ such that $i\in J_\infty$. In this case, 
	$$\underset{\alpha\uparrow1}{\lim}\frac{\mathbb{E}[(X_i-x_i)_+]}{x_1\bar{F}_{X_1}(x_1)}=\underset{\alpha\uparrow1}{\lim}\frac{\bar{F}_{X_i}(x_i)}{\bar{F}_{X_1}(x_i)}\frac{x_i\bar{F}_{X_1}(x_i)}{x_1\bar{F}_{X_1}(x_1)}\frac{\mathbb{E}[(X_i-x_i)_+]}{x_i\bar{F}_{X_i}(x_i)}=0\/,$$
	because $\underset{\alpha\uparrow1}{\lim}\frac{\mathbb{E}[(X_i-x_i)_+]}{x_i\bar{F}_{X_i}(x_i)}=\frac{1}{\theta-1}$, $\underset{\alpha\uparrow1}{\lim}\frac{\bar{F}_{X_i}(x_i)}{\bar{F}_{X_1}(x_i)}=0$, and by Lemma \ref{Lemme+inf} ($X_i=X_j=X_1$) $\underset{\alpha\uparrow1}{\lim}\frac{x_i\bar{F}_{X_1}(x_i)}{x_1\bar{F}_{X_1}(x_1)}=0$.\\ 
	On another hand, for all $j\in\{1,\ldots,d\}\backslash\{i\}$, 
	$$\frac{\mathbb{E}[(X_j-x_j)_+1\!\!1_{\{X_i>x_i\}}]}{x_1\bar{F}_{X_1}(x_1)}=\intevariable{x_j}{+\infty}{\frac{\mathbb{P}(X_j>t,X_i>x_i)}{x_1\mathbb{P}(X_1>x_1)}dt}\/,$$
	so if $j\in J_{C}$, then
	$$\underset{\alpha\uparrow1}{\lim}\frac{\mathbb{E}[(X_j-x_j)_+1\!\!1_{\{X_i>x_i\}}]}{x_1\bar{F}_{X_1}(x_1)}=\underset{\alpha\uparrow1}{\lim}\intevariable{\frac{x_j}{x_1}}{+\infty}{\frac{\mathbb{P}(X_j>tx_1,X_i>x_i)}{\mathbb{P}(X_1>x_1)}dt}=\underset{\alpha\uparrow1}{\lim}\intevariable{\beta_j}{+\infty}{\frac{\mathbb{P}(X_j>tx_1,X_i>x_i)}{\mathbb{P}(X_1>x_1)}dt}\/,$$
	because $\frac{\mathbb{P}(X_j>tx_1,X_i>x_i)}{\mathbb{P}(X_1>x_1)}\leq \frac{\mathbb{P}(X_j>tx_1)}{\mathbb{P}(X_1>x_1)}$ and $\underset{\alpha\uparrow1}{\lim}\frac{\mathbb{P}(X_j>tx_1)}{\mathbb{P}(X_1>x_1)}=0$ for all $t>0$.
	We apply the dominated convergence Theorem, using Potter's bounds associated to $\bar{F}_{X_1}$, to get
	$$\underset{\alpha\uparrow1}{\lim}\frac{\mathbb{E}[(X_j-x_j)_+1\!\!1_{\{X_i>x_i\}}]}{x_1\bar{F}_{X_1}(x_1)}=\intevariable{\beta_j}{+\infty}{\underset{\alpha\uparrow1}{\lim}\frac{\mathbb{P}(X_j>tx_1,X_i>x_i)}{\mathbb{P}(X_1>x_1)}dt}\/,$$
	and since
	$$\frac{\mathbb{P}(X_j>tx_1,X_i>x_i)}{\mathbb{P}(X_1>x_1)}\leq\frac{\mathbb{P}(X_i>x_i)}{\mathbb{P}(X_1>x_1)}=\frac{x_1}{x_i}\frac{x_i\bar{F}_{X_1}(x_i)}{x_1\bar{F}_{X_1}(x_1)}\frac{\bar{F}_{X_i}(x_i)}{\bar{F}_{X_1}(x_i)}\/,$$
	so, by Lemma \ref{Lemme+inf}
	$$\underset{\alpha\uparrow1}{\lim}\frac{\mathbb{P}(X_j>tx_1,X_i>x_i)}{\mathbb{P}(X_1>x_1)}=\underset{\alpha\uparrow1}{\lim}\frac{\mathbb{P}(X_i>x_i)}{\mathbb{P}(X_1>x_1)}=0\/,$$
	we deduce finally that 
	$$\underset{\alpha\uparrow1}{\lim}\frac{\mathbb{E}[(X_j-x_j)_+1\!\!1_{\{X_i>x_i\}}]}{x_1\bar{F}_{X_1}(x_1)}=0, ~\forall j\in J_{C}\/.$$
	If $j\in J_{\infty}\backslash\{i\}$, then  
	\begin{align*}
	\frac{\mathbb{E}[(X_j-x_j)_+1\!\!1_{\{X_i>x_i\}}]}{x_1\bar{F}_{X_1}(x_1)}&=\intevariable{x_j}{+\infty}{\frac{\mathbb{P}(X_j>t,X_i>x_i)}{x_1\mathbb{P}(X_1>x_1)}dt}\\
	&\leq \intevariable{x_1}{+\infty}{\frac{\mathbb{P}(X_j>t,X_i>x_i)}{x_1\mathbb{P}(X_1>x_1)}dt}=\intevariable{1}{+\infty}{\frac{\mathbb{P}(X_j>tx_1,X_i>x_i)}{\mathbb{P}(X_1>x_1)}dt}\/,
	\end{align*}
	we show in the same way as in the previous case that 
	$$\underset{\alpha\uparrow1}{\lim}\intevariable{1}{+\infty}{\frac{\mathbb{P}(X_j>tx_1,X_i>x_i)}{\mathbb{P}(X_1>x_1)}dt}=\intevariable{1}{+\infty}{\underset{\alpha\uparrow1}{\lim}\frac{\mathbb{P}(X_j>tx_1,X_i>x_i)}{\mathbb{P}(X_1>x_1)}dt}=0\/,$$ 
	then  $$\underset{\alpha\uparrow1}{\lim}\frac{\mathbb{E}[(X_j-x_j)_+1\!\!1_{\{X_i>x_i\}}]}{x_1\bar{F}_{X_1}(x_1)}=0, ~~\forall j\in J_{\infty}\backslash\{i\}\/.$$
	If $j\in J_{0}$, then 
	\begin{align*}
	\frac{\mathbb{E}[(X_j-x_j)_+1\!\!1_{\{X_i>x_i\}}]}{x_1\bar{F}_{X_1}(x_1)}&=\intevariable{x_j}{x_1}{\frac{\mathbb{P}(X_j>tx_1,X_i>x_i)}{x_1\mathbb{P}(X_1>x_1)}dt}+\intevariable{x_1}{\infty}{\frac{\mathbb{P}(X_j>tx_1,X_i>x_i)}{x_1\mathbb{P}(X_1>x_1)}dt}\\
	&\leq\frac{x_1-x_j}{x_1}\frac{\bar{F}_{X_i}(x_i)}{\bar{F}_{X_1}(x_1)}+\intevariable{1}{+\infty}{\frac{\mathbb{P}(X_j>tx_1,X_i>x_i)}{\mathbb{P}(X_1>x_1)}dt}\/,
	\end{align*}
	since $\underset{\alpha\uparrow1}{\lim}\intevariable{1}{+\infty}{\frac{\mathbb{P}(X_j>tx_1,X_i>x_i)}{\mathbb{P}(X_1>x_1)}dt}=0$, so, by Lemma \ref{Lemme+inf}, we get
	$$\underset{\alpha\uparrow1}{\lim}\frac{x_1-x_j}{x_1}\frac{\bar{F}_{X_i}(x_i)}{\bar{F}_{X_1}(x_1)}=\underset{\alpha\uparrow1}{\lim}\frac{x_1-x_j}{x_1}\frac{\bar{F}_{X_i}(x_i)}{\bar{F}_{X_1}(x_i)}\frac{x_i\bar{F}_{X_1}(x_i)}{x_1\bar{F}_{X_1}(x_1)}\frac{x_1}{x_i}=0\/,$$
	we obtain from that
	$$\underset{\alpha\uparrow1}{\lim}\frac{\mathbb{E}[(X_j-x_j)_+1\!\!1_{\{X_i>x_i\}}]}{x_1\bar{F}_{X_1}(x_1)}=0,~~\forall j\in J_{0}\/, $$
	and consequently 
	$$\underset{\alpha\uparrow1}{\lim}\frac{\mathbb{E}[(X_j-x_j)_+1\!\!1_{\{X_i>x_i\}}]}{x_1\bar{F}_{X_1}(x_1)}=0,~~\forall j\in \{1,\ldots,d\}\backslash\{i\}\/.$$
	The $i^{\mbox{th}}$ equation of System \ref{eq1} divided by $x_1\bar{F}_{X_1}(x_1)$ can be written in the form 
	\begin{align*}
	\frac{(2\alpha-1)\mathbb{E}[(X_i-x_i)_+]}{x_1\bar{F}_{X_1}(x_1)}-\frac{1-\alpha}{\bar{F}_{X_1}(x_1)}\frac{x_i-\mathbb{E}[X_i]}{x_1}&=\sum_{\underset{j\neq i}{j=1}}^{d}\frac{(1-\alpha)\mathbb{E}[(X_j-x_j)_-1\!\!1_{\{X_i<x_i\}}]}{x_1\bar{F}_{X_1}(x_1)}\\&~~-\sum_{\underset{j\neq i}{j=1}}^{d}\frac{\alpha\mathbb{E}[(X_j-x_j)_+1\!\!1_{\{X_i>x_i\}}]}{x_1\bar{F}_{X_1}(x_1)}\/,
	\end{align*}
	going through the limit $(\alpha\rightarrow 1)$ in this equation leads to
	$$-\underset{\alpha\uparrow1}{\lim}\frac{1-\alpha}{\bar{F}_{X_1}(x_1)}\frac{x_i}{x_1}=\underset{\alpha\uparrow1}{\lim}\frac{1-\alpha}{\bar{F}_{X_1}(x_1)}\sum_{\underset{j\neq i}{j=1}}^{d}\frac{x_j}{x_1}\/,$$
	which is possible only if $\underset{\alpha\uparrow1}{\lim}\frac{1-\alpha}{\bar{F}_{X_1}(x_1)}\frac{x_i}{x_1}=0$, and that is contradictory with Equation \ref{DAsy-eq2}.
\end{proof}

\end{document}